 \documentclass[journal]{IEEEtran}
 \newcommand{\SortNoop}[1]{}
 \usepackage{amssymb,stmaryrd,amsmath,amsfonts,rotating,mathrsfs,TIKZ}
 \usepackage{MnSymbol}
 
 \usepackage[noadjust]{cite}
 \usepackage{color}
 \usepackage{epic}
 \RequirePackage{bbm}
\newtheorem{theorem}{Theorem}
\newtheorem{assumption}{Assumption}
\newtheorem{question}{Question}

\newcommand{\btheo}{\begin{theorem}}
\newcommand{\etheo}{\end{theorem}}
\newcommand{\bproof}{\begin{proof}}
\newcommand{\eproof}{\end{proof}}
\newtheorem{definition}[theorem]{Definition}
\newcommand{\bdefi}{\begin{definition}}
\newcommand{\edefi}{\end{definition}}
\newtheorem{fact}[theorem]{Fact}
\newcommand{\bprop}{\begin{fact}}
\newcommand{\eprop}{\end{fact}}
\newtheorem{corollary}[theorem]{Corollary}
\newcommand{\bcor}{\begin{corollary}}
\newcommand{\ecor}{\end{corollary}}
\newtheorem{example}[theorem]{Example}
\newcommand{\bex}{\begin{example}}
\newcommand{\eex}{\end{example}}
\newtheorem{lemma}[theorem]{Lemma}
\newcommand{\blemma}{\begin{lemma}}
\newcommand{\elemma}{\end{lemma}}
\newtheorem{remark}[theorem]{Remark}
\newcommand{\bremark}{\begin{remark}}
\newcommand{\eremark}{\end{remark}}
\newtheorem{conj}[theorem]{Conjecture}
\newcommand{\bconj}{\begin{conj}}
\newcommand{\econj}{\end{conj}}

\newcommand{\naturals}{\ensuremath{\mathbb{N}}}

\def\0{{\tt 0}} 
\def\1{{\tt 1}} 
\def\?{{\tt *}} 
\renewcommand{\mid}{\,|\,}

\newcommand{\EEx}{\hfill $\Diamond$}

 \allowdisplaybreaks
 
\begin{document}
 \title{Finite-Length Scaling for Polar Codes}

\author{S.~Hamed~Hassani, Kasra Alishahi, and~R\"{u}diger~Urbanke%
\thanks{This work was supported in part by grant No. $200020_146832/1$ and grant No. $200021-121903$ of the Swiss National Science Foundation. This paper was presented in part in \cite{HKU10}, \cite{goli}.
S. H. Hassani is with the Department of Computer Science,
ETHZ, 8092 Zurich, Switzerland
(e-mail: hamed@inf.ethz.ch).}%
\thanks{K. Alishahi is with the department of Mathematical Sciences, Sharif University of Technology, Tehran, Iran
(email:  alishahi@sharif.edu).}%
\thanks{R. Urbanke is with the School of Computer and Communication Science,
EPFL, 1015 Lausanne, Switzerland
(e-mail: rudiger.urbanke@epfl.ch).}
}
 \maketitle
 \begin{abstract}

Consider a binary-input memoryless output-symmetric channel $W$.
Such a channel has a capacity, call it $I(W)$, and for any $R<I(W)$
and strictly positive constant $P_{\rm e}$ we know that we can
construct a coding scheme that allows transmission at rate $R$ with
an error probability not exceeding $P_{\rm e}$.  Assume now that
we let the rate $R$ tend to $I(W)$ and we ask how we have to ``scale''
the blocklength $N$ in order to keep the error probability fixed
to $P_{\rm e}$.  We refer to this as the ``finite-length scaling''
behavior.  This question was addressed by Strassen as well as
Polyanskiy, Poor and Verdu, and the result is that $N$ must grow
at least as the square of the reciprocal of $I(W)-R$.

Polar codes are optimal in the sense that they achieve capacity.
In this paper, we are asking to what degree they are also optimal in terms of their
finite-length behavior.  Since the exact
scaling behavior depends on the choice of the channel our objective
is to provide scaling laws that hold universally for all
binary-input memoryless output-symmetric channels.
 Our approach is based on analyzing the dynamics of the
un-polarized channels. More precisely, we provide bounds on (the
exponent of) the number of sub-channels whose Bhattacharyya constant
falls in a fixed interval $[a,b]$.  Mathematically, this can be
stated as bounding the sequence $\bigl\{\frac{1}{n} \log {\rm{Pr}}(Z_n
\in [a,b])\bigr\}_{n \in \mathbb{N}}$, where $Z_n$ is the Bhattacharyya
process.  We then use these bounds to derive trade-offs between the
rate and the block-length.

The main results of this paper can be summarized as follows.  Consider the sum of
Bhattacharyya parameters of sub-channels chosen (by the polar coding
scheme) to transmit information.  If we require this sum to be
smaller than a given value $P_{\rm e}>0$, then the required
block-length $N$ scales in terms of the rate $R < I(W)$ as $N \geq
\frac{\alpha}{(I(W)-R)^{\underline{\mu}}}$, where $\alpha$ is a
positive constant that depends on $P_{\rm e}$ and $I(W)$.  We show
that $\underline{\mu} = 3.579$ is a valid choice, and we conjecture
that indeed the value of $\underline{\mu}$ can be improved to
$\underline{\mu}=3.627$, the parameter for the binary erasure
channel.  Also, we show that with the same requirement on the sum
of Bhattacharyya parameters, the block-length scales in terms of
the rate  like $N \leq \frac{\beta}{(I(W)-R)^{\overline{\mu}}}$,
where $\beta$ is a constant that depends on $P_{\rm e}$ and $I(W)$,
and $\overline{\mu}=6$.
 \end{abstract}

\section{Introduction}
Polar coding schemes \cite{Ari09} provably achieve the capacity of
a wide class of channels including binary-input memoryless output-symmetric (BMS)
channels.

In coding, the three most important parameters are:   rate ($R$),
block-length ($N$), and  block error probability ($P_{\rm e}$).
Ideally, given a family of codes such as the family of polar codes,
one would like to be able to describe the exact relationship between
these three parameters. This however is a formidable task.  It is slightly
easier to fix one of the parameters and then to describe the
relationship (scaling) of the remaining two.

For example, assume that we fix the rate and consider the relationship
between the error probability and the block-length.  This is the
study of the classical error exponent.  For instance, for random
codes a closer look shows that $P_{\rm e}=e^{-N E(R,W) + o(N)}$,
where $E(R,W)$ is the so-called {\em random coding error exponent} \cite{Gal}
of the channel $W$.  For polar codes, Ar\i kan and Telatar \cite{AT09}
showed that when $W$ is a BMS channel, for any fixed rate $R <I(W)$ the
block error probability of polar codes with the successive cancellation (SC) decoder is upper bounded 
by $2^{-N^{\beta}}$ for
any $\beta<\frac 12$ and $N$ large enough.  This result was refined
later in \cite{HMTU} to be dependent on $R$, i.e. for polar codes
with the SC decoder \begin{equation*}
P_{\rm e}= 2^{-2^{\frac{n}{2} + \sqrt{n} Q^{-1} (\frac{R}{I(W)}) +
o(\sqrt{n})}}, \end{equation*} where\footnote{In this paper all  the
logarithms are in base $2$.} $n=\log N$ and $Q(t) \triangleq
\int_{t}^\infty e^{-z^2 / 2} dz/{\sqrt{2 \pi}}$.

Another option is to fix the error probability and to consider the
relationship between the block-length and the rate.  In other words,
given a code and a desired (and fixed) error probability $P_{\rm
e}$, what is the block-length $N$ required, in terms of the rate
$R$, so that the code has error probability less than $P_{\rm e}$?
This scaling is arguably more relevant (than the error exponent)
from a practical point of view since we typically have a certain
requirement on the error probability and then are interested in
 using the shortest code possible to transmit at a certain rate. 

As a benchmark, let us mention what is the shortest block-length
that we can hope for.  Some thought clarifies that the random
variations of the channel itself require $R \leq I(W)-
\Theta(\frac{1}{\sqrt{N}})$ or equivalently $N \geq
\Theta(\frac{1}{(I(W)-R)^2})$.  Indeed, a sequence of works starting
from \cite{Dob}, then \cite{Str}, and finally  \cite{PPV} showed
that the minimum possible block-length $N$ required to achieve a
rate $R$ with a fixed error probability $P_{\rm e}$ is roughly equal
to
\begin{equation} \label{min_block}
 N \approx \frac{V( Q^{-1}(P_{\rm e}))^2 }{(I(W)-R)^2},
\end{equation}
where $V$ is a characteristic of the channel referred to as channel
dispersion.  In other words, the best codes require a block-length
of order $\Theta(\frac{1}{(I(W)-R)^2})$.

The main objective of this paper is to characterize similar types of
relations for polar codes with the SC decoder.  We argue in this
paper that this problem is fundamentally related to the dynamics
of channel polarization and especially the speed of which the
polarization phenomenon is taking place. We then provide analytical 
bounds on the speed of polarization for BMS channels. Finally,
by using these bounds we derive scaling laws between the block-length and the 
rate (given a fixed error probability) that 
hold universally for all BMS channels.  To state things in a more convenient
language, let us begin by reviewing some conventional definitions, settings, and results
regarding polarization and polar codes.

 %
 %

 \subsection{Preliminaries}
 Let $W: \mathcal{X} \to \mathcal{Y}$ be a BMS channel, with input alphabet $\mathcal{X}
 =\{0,1\}$,  output alphabet\footnote{Throughout this paper we assume for simplicity that the output alphabet of the channel is discrete. 
 However, the results can be naturally extended to channels with continuous alphabet.}  $\mathcal{Y}$, and the transition probabilities $\{W(y \mid x): x \in \mathcal{X}, y \in \mathcal{Y}\}$. 
  We consider the following three parameters for the channel $W$
 \begin{align}
 & H(W) =\! \sum_{y \in \mathcal{Y}}W(y \mid 1) \log \frac{W(y \mid 1) + W(y\mid 0)}{W(y \mid 1)}, 
  \label{H(W)} \\
 & Z(W) =\! \sum_{y \in \mathcal{Y}} \sqrt{W(y\mid 0)W(y \mid 1)},  \label{Z(W)}\\
&  E(W) =\! \sum_{y \in \mathcal{Y}} \! W(y | 1) \bigl( \mathbbm{1}_{\bigl \{W(y | 0) > W(y | 1) \bigr \}} \! + \! \frac 12 \mathbbm{1}_{\bigl \{W(y | 0) = W(y | 1) \bigr \}} \bigr ), \label{E(W)}
 \end{align}
 where $\mathbbm{1}_{\{A\}}$ is equal to $1$ if $A$ is true and $0$ otherwise.  The parameter $H(W)$ is equal to the entropy of the input 
 of $W$ given its output when we assume  uniform distribution on 
 the inputs, i.e.,  $H(W)=H(X\mid Y)$. Hence, we call the parameter $H(W)$ the
  entropy of the channel $W$.  Also note that the capacity of $W$, which we denote by $I(W)$, is given by $I(W)=1-H(W)$.  
 The parameter $Z(W)$ is called the Bhattacharyya parameter of $W$ and $E(W)$ is called the error probability of $W$. It can be shown that $E(W)$ is equal to the error probability 
 in estimating the channel input $x$ on the basis of 
 the channel output $y$ via the maximum-likelihood 
 decoding of $W(y|x)$ (with the further assumption 
 that the input has uniform distribution). The following relations hold between these parameters
  (see for e.g., \cite{Ari09} and\footnote{One way to prove all these inequalities is by using an equivalent representation of BMS channels as probability distributions on the uniform interval (\cite[Section 4.1.4]{RiU08}). 
  Speaking very briefly, any BMS channel $W$ can be represented by a density $a_W (x)$ where $x \in [0,1]$. In this setting, the equivalent definitions of the parameters $H(W), Z(W)$ and $E(W)$ are as follows: $H(W) = \int_{0}^1 h_2(\frac{1-x}{2}) a_W(x) dx$, $Z(W) = \int_{0}^1 \sqrt{1-x^2} a_W(x) dx$, and $E(W) = \int_{0}^1 \frac{1-x}{2} a_W(x) dx$. 
  Now, by using these new definitions, the relation \eqref{bounds1} is easy to prove by comparing the corresponding kernels of the integrals. 
 Relation \eqref{bounds2} follows in the same way and by further noting that the function $h_2(x)$ is concave. More precisely, we can write 
 $H(W) = \int_{0}^1 h_2(\frac{1-x}{2}) a_W(x) dx \leq h_2( \int_{0}^1 \frac{1-x}{2} a_W(x) dx ) = h_2(E(W))$. Relations \eqref{bounds3} and \eqref{EZ_bound} 
 also follow in the same manner. } \cite[Chapter 4]{RiU08}):
 \begin{align} \label{bounds1}
 & 0 \leq 2 E(W) \leq H(W) \leq Z(W) \leq 1,\\
 \label{bounds2}
 & H(W)\leq h_2(E(W)),\\ \label{bounds3}
  &  Z(W) \leq \sqrt{1-(1-H(W))^2}\\
  & 2E(W) \geq 1-\sqrt{1-Z(W)^2}, \label{EZ_bound}
 \end{align}
 where $h_2(\cdot)$ denotes the binary entropy function, i.e.,
 \begin{equation}
  h_2(x)=-x \log (x)-(1-x)\log (1-x).
 \end{equation}

 \subsection{Channel Transform} \label{chantrans}
 Let $\mathcal{W}$ denote the set of all BMS channels and consider a transform $W \to (W^0, W^1)$ that maps $\mathcal{W}$ to $\mathcal{W}^2$ in the following manner.
 Having the channel $W: \{0,1\} \to \cal Y$,  the channels $W^0: \{0,1\} \to {\cal Y}^2$ and $W^1: \{0,1\} \to \{0,1\} \times  {\cal Y}^2$  are
 defined as
 \begin{align} \label{1}
 & W^0(y_1,y_2 | x_1)= \sum_{x_2 \in \{0,1\} } \frac 12 W(y_1| x_1 \oplus x_2) W(y_2|x_2)  \\ \label{2}
 & W^1(y_1,y_2,x_1 | x_2)= \frac 12 W(y_1 | x_1 \oplus x_2) W(y_2 | x_2),
 \end{align}
 The transform  $W \to (W^0, W^1)$ is also known as the channel splitting transform. A direct consequence of the chain rule of entropy yields
 \begin{equation} \label{I_preserve}
 \frac{H(W^0)+ H(W^1)}{2}= H(W).
 \end{equation}
 Regarding the other parameters, we have (see \cite{Ari09} and\footnote{More precisely, we refer to \cite[Theorem 4.141]{RiU08} as well as \cite[Exercise 4.62]{RiU08}.} \cite[Chapter 4]{RiU08})
 \begin{align}
 &  Z(W) \sqrt{2-Z(W)^2} \leq  Z(W^0) \leq 1-(1-Z(W))^2, \label{Zgen-}\\
 & Z(W^1)=Z(W)^2, \label{Zgen+}
 \end{align}
 and (see\footnote{See the previous footnote.} \cite[Chapter 4]{RiU08})
 \begin{align}
 & E(W^0)=1-(1-E(W))^2,\\
 & E(W)^2 \leq E(W^1) \leq E(W).
 \end{align}
 
\subsection{Channel Polarization} \label{chanpol}
Consider an infinite binary tree with the root node placed at the top. In this tree each vertex has $2$ children and 
there are $ 2^n$ vertices at level $n$.
Assume that we label these vertices from left to right from $0$ to $
2^n-1$. Here, we intend to assign to each vertex of the tree a BMS channel.
We do this by a recursive procedure. Assign to the root node the channel $W$
itself. Now consider the channel splitting transform $W \to (W^0, W^1)$  
and from left to right, assign $W^0$ and $W^1$ to the children
of the root node. In general, if $Q$ is the channel that is assigned to vertex $v$, 
we assign $Q^0$ to the ``left" child of $v$ and $Q^1$ to the ``right" child of $v$. 
In this way, we recursively assign a channel to all the 
vertices of the tree. Figure \ref{SKA2}\begin{figure}[ht!]
\centering
\begin{tikzpicture}[scale=.8]
\tikzstyle{every node}=[draw,shape=circle]; 

\draw(4cm,4cm) node[draw=none](x) {$W$};

\draw(1.5cm,3cm) node[draw=none] (x0){$W^0$};
\draw(6.5cm,3cm) node[draw=none] (x1){$W^1$};

\draw(-.3cm,1cm) node[draw=none] (x00){$(W^0)^0$};
\draw(3.3,1cm) node[draw=none] (x01){$(W^0)^1$};
\draw(4.7cm,1cm) node[draw=none] (x10){$(W^1)^0$};
\draw(8.3,1cm) node[draw=none] (x11){$(W^1)^1$};

\draw[-] (x) to node[draw=none,above]{}(x0);
\draw[-] (x) to node[draw=none,above]{}(x1);

\draw[-] (x0) to node[draw=none,above]{}(x00);
\draw[-] (x0) to node[draw=none,above]{}(x01);

\draw[-] (x1) to node[draw=none,above]{}(x10);
\draw[-] (x1) to node[draw=none,above]{}(x11);

\draw(8.8cm,.5cm) node[draw=none] (x001){$\ddots$};
\draw(4.7cm,.5cm) node[draw=none] (x001){$\vdots$};
\draw(3.3cm,.5cm) node[draw=none] (x001){$\vdots$};
\draw(-.5cm,.5cm) node[draw=none] (x001){$\udots$};
\end{tikzpicture}
\caption{The infinite binary tree and the channels assigned to it.}
\label{SKA2}
\end{figure}
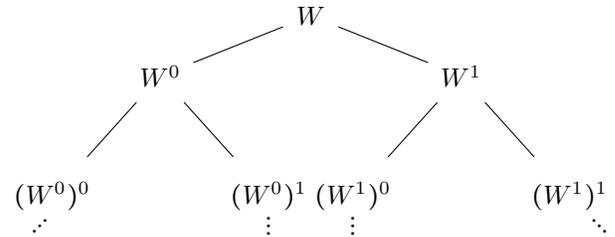 shows the first $2$ levels 
of the binary tree.
Assuming $N=2^n$, we let $W_{N}^{(i)}$ denote the channel that is assigned to a vertex with label 
$i$ at level $n$ of the tree, $0
\leq i \leq N-1$. As a result, one can equivalently relate 
the channel $W_{ N} ^{(i)}$ to $W$ via the following procedure: let the binary representation
of $i$ be $b_1 b_2 \cdots b_n$, where $b_1$ is the most significant
digit. Then we have
\begin{equation*}
 W_{ N} ^{(i)}=(((W^{b_1})^{b_2})^{\cdots})^{b_n}.
\end{equation*}
As an example, assuming $i=6$, $n=3$ we have $W_{8} ^{(6)} =((W^{1})^{1})^{0}$.  
We now proceed with defining a stochastic process called the polarization process. 
This process can be considered as a stochastic representation of the channels associated 
to different levels of the infinite binary tree.  
 
\subsection{Polarization Process} \label{polproc}
Let $\{B_n, n\geq 1\}$ be a sequence of independent and identically distributed (iid) Bernoulli($\frac12$) 
random variables. Denote by $(\mathcal{F}, \Omega, {\rm{Pr}})$ the 
probability space generated by 
this sequence and let  $(\mathcal{F}_n, \Omega_n, {\rm{Pr}}_n)$ be the probability 
space generated by $(B_1, \cdots,B_n)$.  For a BMS channel $W$, define a 
random sequence of channels $W_n$, $n \in \mathbb{N} \triangleq \{0,1,2,\cdots\}$, 
as $W_0=W$ and
\begin{equation} \label{W_n}
W_{n}=  \left\{
\begin{array}{lr}
W_{n-1} ^{0} &  \text{if $B_n=0$},\\
W_{n-1} ^{1} &   \text{if $B_n=1$},
\end{array} \right.
\end{equation}
where the channels on the right side are given by the transform $W_{n-1}\to (W_{n-1}^0, W_{n-1}^1)$. 
Let us also define the random processes $\{H_n\}_{n \in \mathbb{N}}$, $\{I_n\}_{n \in \mathbb{N}}$, $\{Z_n\}_{n \in \mathbb{N}}$ and $\{E_n\}_{n \in \mathbb{N}}$ as
$H_n=H(W_n)$, $I_n=I(W_n)=1-H(W_n)$, $Z_n=Z(W_n)$ and $E_n=E(W_n)$. 
\begin{example}
 By a straightforward calculation one can show that for $W=\text{BEC}(z)$ we have
\begin{align}
&W^0={\rm{BEC}}(1-(1-z)^2)\\
& W^1={\rm{BEC}}(z^2).
\end{align}
Hence, when $W=\text{BEC}(z)$, the channel $W_n$ is always a BEC. Furthermore, the processes $H_n, I_n,Z_n$ and $E_n$ admit simple closed form recursions as follows. We have 
$H_0=z$ and for $n \geq 1$
\begin{equation} \label{H_n}
H_{n} =\left\{ \begin{array}{cc} 
1-(1-H_{n-1})^2,&\text{ w.p. }\frac12\\
H_{n-1}^2,&\text{ w.p. }\frac12.
\end{array}\right.
\end{equation}
Also, we have\footnote{For the channel $W=\text{BEC}(z)$, it is easy to show that $2E(W)=H(W)=Z(W)=z$.} $2E_n=H_n=1-I_n=Z_n$.

For channels other than the BEC, the channel $W_n$ gets quite complicated
in the sense that  the cardinality
of the output alphabet of the channel $W_n$ is doubly exponential in $n$ (or exponential in $N$). 
Thus, tracking the exact outcome of $W_n$ seems 
to be a difficult task 
(for more details see \cite{TV10, RHTT}). Instead, as we will see in the sequel, 
one can prove many interesting properties regarding the processes $H_n,Z_n$ and $E_n$.  
\end{example}

Let us quickly review the limiting properties of the above mentioned processes \cite{Ari09, AT09}. 
From \eqref{I_preserve} and \eqref{W_n}, one can write for $n \geq 1$
\begin{equation}
 \mathbb{E}[H(W_n) \mid W_{n-1}]\stackrel{\eqref{W_n}}{=}\frac{H(W_{n-1}^0)+H(W_{n-1}^1)}{2}\stackrel{\eqref{I_preserve}}{=}H(W_{n-1}).
\end{equation}
Hence, the process $H_n$ is a martingale. Furthermore, since $H_n$ is also bounded (see \eqref{bounds1}), by 
Doob's martingale convergence theorem, the process $H_n$ converges almost 
surely to a limit random variable $H_\infty$. As $H_n$ is also bounded,  we have for $n \to \infty$
\begin{equation*}
\mathbb{E}\bigl[ | H_{n}-H_{n-1}| \bigr]= \mathbb{E}\bigl[ |H(W_n^0)-H(W_n)| \bigr] \to 0.
\end{equation*}
As a result, we must have that $H(W_n^0)-H(W_n)$ converges to $0$ almost surely (a.s.). We will shortly prove that for a channel $P$, in order to have $H(P^0) \approx H(P)$ we must either have $H(P) \approx 0$ 
(i.e., $P$ is the noiseless channel) or $H(P) \approx 1$ (i.e., $P$ is the completely noisy channel). By this claim and the fact that $H_n$ converges a.s. to $H_\infty$, we conclude that $H_\infty$ 
takes its values in the set $\{0,1\}$. Also, as $\mathbb{E}[H_n]=\mathbb{E}[H_\infty]=H(W)$, we obtain
\begin{equation} \label{H_inf}
H_{\infty}=  \left\{
\begin{array}{lr}
0 &  \text{w.p. $1-H(W)$},\\
1  &   \text{w.p. $H(W)$}.
\end{array} \right.
\end{equation}
It remains to prove the claim mentioned above. It is sufficient\footnote{Here, we are skipping some unnecessary details. For the sake of completeness, we note that the function $H(\cdot)$ is a continuous function over the space of BMS channels. For more details, we refer to \cite[Chapter 4]{RiU08}.} to show that for a channel $P$, in order to have $H(P^0)=H(P)$ we must have $H(P) \in \{0,1\}$. We use the so called \textit{extremes of information combining} inequalities \cite[Theorem 4.141]{RiU08}:  Let $P$ be an arbitrary BMS channel. To simplify
notation, let $h \triangleq H(P)$ and also let $\epsilon \in [0, \frac 12]$ be such that $h_2(\epsilon)=h$ (in this way, the two channels BEC($h$) and BSC($\epsilon$) have the same capacity). We have
\begin{align} 
& \label{ex-}  h   \leq \overbrace {H(\text{BSC}(\epsilon)^0 )}^{h_2(2\epsilon(1-\epsilon))}  \leq  H(P^0) 
 \leq  \overbrace{H(\text{BEC}(h)^0)}^{ 1-(1-h)^2}, \\
& \label{ex+}  \underbrace{H(\text{BEC}(h)^1 )}_{h^2} \leq H(P^1) \leq \underbrace{ H(\text{BSC}(\epsilon)^1)}_{2h-h_2(2 \epsilon(1-\epsilon))} \leq h. 
\end{align}
Now, to prove the claim, assume that $P$ is such that $H(P^0)=H(P)=h$. Using \eqref{ex-} we obtain $H(\text{BSC}(h)^0)=H(P)$ or equivalently
$h_2(2 \epsilon(1-\epsilon))=h=h_2(\epsilon)$. As a result, $\epsilon$ must be a solution of the equation $\epsilon=2\epsilon(1-\epsilon)$ which yields $\epsilon \in \{0,\frac 12\}$. 
Also, as 
$H(P)=h_2(\epsilon)$, then $H(P)$ can either be $0$ or $1$ and hence the claim is justified. 
Using the bounds \eqref{bounds1}-\eqref{bounds3} it is clear that the processes $Z_n$ and $E_n$ converge a.s. to $H_\infty$ and 
$\frac 12 H_\infty$, respectively.    
 \subsection{Polar Codes}
Given the rate $R<I(W)$, polar coding  is based on selecting a 
set of $2^nR$ rows of 
the matrix 
$G_n= \bigl [ \begin{smallmatrix} 1 & 0 \\ 1 &1  \end{smallmatrix} \bigr]^{\otimes n}$ 
to form a $2^nR \times 2^n$ matrix 
which is used as the generator matrix in the 
encoding 
procedure. The way this set is selected is dependent on the channel $W$ and is briefly explained as follows:  
Order the the set of channels  $\{W_{N}^{(i)}\}_{0 \leq i \leq N-1}$ according to their error probability (given in \eqref{E(W)}). Then, pick the $N\cdot R$ channels which have the smallest error probability and consider the rows of $G_n$ with the same indices as these
channels.\footnote{One can also construct polar codes by choosing the channels that have the least Bhattacharyya parameter or the least entropy (see \eqref{H(W)} and \eqref{Z(W)}). In essence, these constructions are all equivalent except that a few indices might be different. Choosing the channels that have the least error probability has the advantage of minimizing the ``union"-type bounds that can be provided on the block-error probability when we use SC decoding (see e.g. the right side of \eqref{P_e}).} 
E.g., if the channel $W_{N}^{(i)}$ is chosen, then 
the $i$-th row of $G_n$ is selected. 
In the following, given $N$, we call the set of indices of $N \cdot R$ channels 
with the least error probability \emph{the set of good indices} and denote it by 
$\mathcal{I}_{N,R}$. 
Moreover, we will frequently use the terms 
``the set of good indices'' and $\mathcal{I}_{N,R}$ interchangeably.

We now briefly explain why such a code construction is reliable for any rate $R<I(W)$, provided that the block-length is large enough. It is proven in \cite{Ari09} that the block error probability of such polar coding scheme
 under SC decoding, denoted by $P_{\rm e}$, is bounded from both sides by\footnote{Note here that by \eqref{E(W)} the
 error probability of a BMS channel is less than its Bhattacharyya value. Hence, the right side of 
\eqref{P_e} is a better upper bound for the block error probability than the sum of Bhattacharyya values.}
\begin{equation} \label{P_e}
\max_{i \in \mathcal{I}_{N,R} } E(W_N^{(i)}) \leq P_{\rm e} \leq \sum_{i \in \mathcal{I}_{N,R}} E(W_{N}^{(i)}).
\end{equation}   
Recall 
from Subsection~\ref{polproc} that the process $E_n=E(W_n)$ converges a.s. to a random variable $E_\infty$ such that ${\rm{Pr}}(E_\infty=0)=I(W)$. 
Hence, it is clear from the definition of the set of good indices, $\mathcal{I}_{N,R}$, that the left side of \eqref{P_e} decays to $0$ for any $R < I(W)$ as $n$ grows large. However, the story is not over yet as this is only a lower bound on $P_{\rm e}$. 
Nonetheless, one can also show that the right side of \eqref{P_e} decays to $0$. This was initially shown in \cite{Ari09}, and later in \cite{AT09} it was proven that all of the three terms in 
\eqref{P_e} behave like $2^{-2^{\frac{n}{2}+ o(n)}}$.

\section{Problem Formulation}  \label{prob_form}
As we have seen in the previous section, the processes $H_n$ and $Z_n$  polarize in the sense that they
converge a.s. to  $\{0,1\}$-valued random variables $H_\infty$ and $Z_\infty$, respectively. 
In other words, almost surely as $n$ grows, the value of $Z_n$ (or $H_n$) 
is either very close to $0$ or very close to $1$. 
Here,  we investigate the dynamics of polarization.
We start by noting that at each time $n$ there still exists a (small and  in $n$ vanishing) 
probability that the process $Z_n$ (or $H_n$) takes
a value far away from the endpoints of the unit interval (i.e., $0$ and $1$). 
Our primary objective is to study these small  probabilities. 
More concretely, let $0 < a < b < 1$ be constants and consider the quantity
${\rm Pr} (Z_n \in [a, b])$. This quantity represents the fraction of
sub-channels that are still un-polarized at time $n$.  An important question is how fast (in terms of $n$) the quantity 
${\rm Pr} (Z_n \in [a, b])$ decays to zero. This question is intimately related to measuring the limiting properties 
of the sequence $ \{ \frac{1}{n} \log {\rm Pr} (Z_n \in [a, b]) \}_{n \in \mathbb{N}} $.  
\begin{example}
Assume $W={\rm{BEC}}(z)$. In this case the process $Z_n$ has a simple closed form recursion as $Z_0=z$ and
\begin{equation} \label{Z_n}
Z_{n+1} =\left\{ \begin{array}{cc} 
Z_{n}^2,&\text{ w.p. }\frac12,\\
1-(1-Z_{n})^2,&\text{ w.p. }\frac12.
\end{array}\right.
\end{equation}
 Hence, it  is  straightforward  to compute the value 
${\rm Pr} (Z_n \in [a, b])$ numerically. Let $a=1-b=0.1$.  Figure~\ref{fig:BEC_log} shows the value $\frac{1}{n} \log({\rm Pr}(Z_n \in [a,b]))$ in terms of 
$n$ for $z=0.5, 0.6, 0.7$. 
\begin{figure}[htb] 
\begin{center} \input{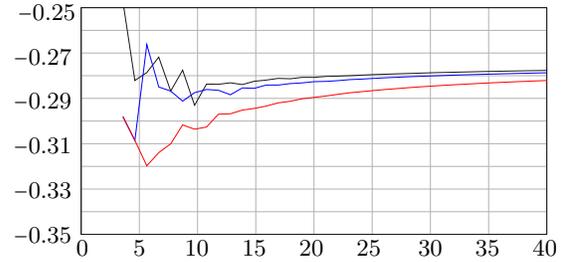} 
\end{center}
\caption{The value of $\frac{1}{n} \log({\rm Pr}(Z_n \in [a,b])$ versus $n$ for $a=1-b=0.1$ when $W$ is a BEC with
erasure probability  $z=0.5$ (top curve), $z=0.6$ (middle curve), and $z=0.7$ (bottom curve).}\label{fig:BEC_log}
\end{figure} 
This figure suggests that the sequence $\{ \frac{1}{n} \log {\rm Pr} (Z_n \in [a, b]) \}_{n \in \naturals}$ converges to a limiting 
value that is somewhere between $-0.27$ and $-0.28$. Note that for different values of $z$, the 
limiting values are very close to each other.
\EEx
\end{example}
For other BMS channels, the process $Z_n$ does not have a simple closed form recursion
as for the BEC, and hence 
we need to use  approximation  methods (for more details see \cite{TV10, RHTT}).  Using such methods, we have plotted in Figure~\ref{fig:general_log} 
the value 
of ${\rm Pr} (Z_n \in [a, b])$ ($a=1-b=0.1$) for the channel families 
BSC($\epsilon$), and BAWGNC($\sigma$) with different parameter values. 
\begin{figure}[htb] 
\begin{center} \input{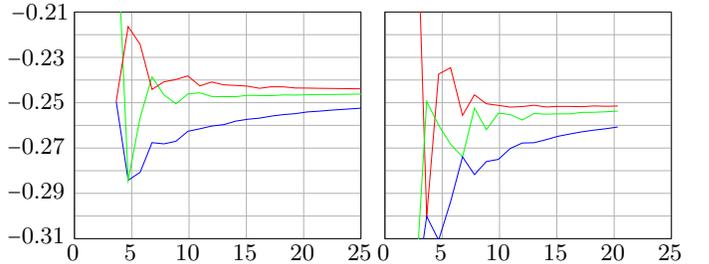} 
\end{center}
\caption{ \emph{Left figure:} The value of $\frac{1}{n} \log({\rm Pr}(Z_n \in [a,b])$ versus $n$ for $a=1-b=0.1$ and $W$ being a BSC with
cross-over probability  $\epsilon=0.11$ (top curve), $\epsilon=0.146$ (middle curve), and $\epsilon=0.189$ (bottom curve).
These BSC channels have capacity $0.5$, $0.4$ and $0.3$, respectively.
\emph{Right figure:} 
the value of $\frac{1}{n} \log({\rm Pr}(Z_n \in [a,b])$ versus $n$ for $a=1-b=0.1$ and $W$ is a BAWGN with
noise variance  $\sigma=0.978$ (top curve), $\sigma=1.149$ (middle curve), and $\sigma= 1.386$ 
(bottom curve). These BAWGN channels have capacities $0.5$, $0.4$ and $0.3$, respectively. 
}\label{fig:general_log}
\end{figure} 

The above numerical evidence suggests that the 
quantity ${\rm Pr} (Z_n \in [a, b])$ decays to zero exponentially 
fast in $n$. Further, we observe that the limiting value of this sequence is  dependent 
on the starting channel $W$ (e.g., from the figures it is clear that the channels BEC, BSC and BAWGN have different limiting values). Let us
 now be concrete and rephrase the above speculations as follows.
\begin{question}\label{Q1}
Does the quantity ${\rm Pr}(Z_n \in [a,b])$ decay exponentially in $n$? If yes,
 what is the limiting value of $\frac{1}{n} \log {\rm Pr}(Z_n \in [a,b])$ and how is this limit 
 related to the starting channel $W$ and the choice of $a$ and $b$?  
\end{question}
From Figures~\ref{fig:BEC_log} and \ref{fig:general_log}, 
we observe that the value of  $\frac{1}{n} \log {\rm Pr}(Z_n \in [a,b])$ is the least when $W$ is a BEC
 and this suggests that the channel BEC polarizes faster than the other BMS channels.  This is intuitively 
justified as follows: Fix a value $z\in (0,1)$ and assume that $W$ is a BMS channel with Bhattacharyya
parameter $Z(W)=z$. Now, consider the values $Z(W^0)$ and $Z(W^1)$. Using relations \eqref{Zgen-} and 
\eqref{Zgen+}, it is clear that the values $Z(W^0)$ and $Z(W^1)$ are closest to the end points of the 
unit interval if $W$ is a BEC. In other words, at the channel splitting transform, the channel ${\rm BEC}(z)$ polarizes faster 
than the other BMS channels.    
\begin{question}\label{Q2}
For which set of channels  does the 
quantity ${\rm Pr}(Z_n \in [a,b])$ decay the fastest or the slowest?
\end{question}
Let us now be more ambitious and aim for the ultimate goal.
\begin{question}\label{Q3}
Can we characterize the 
exact behavior of  ${\rm Pr}(Z_n \in [a,b])$ as a function of $n$, $a,b$ and $W$? 
\end{question}
Finally, we ask how the answers to the above questions will guide 
us through the understanding of the finite-length scaling behavior of polar codes. 
An immediate relation stems from the fact that the 
quantity ${\rm Pr}(Z_n \in [a,b])$  indicates 
the portion of the sub-channels that have not polarized at time $n$. In particular, all the 
channels in this set  have a large Bhattacharyya value (and hence a large error probability).  
Consequently, if any of such un-polarized channels (or equivalently indices) are included in the set of good indices then the error probability would not be small (see \eqref{P_e}). 
Thus, the maximum reliable 
rate that we can achieve is restricted by the portion of these yet un-polarized channels.  
The answers to  Questions \ref{Q1}-\ref{Q3} posed above will therefore be crucial in finding answers to the following question.
\begin{question}\label{Q4}
Fix the channel $W$ and a target block error probability $P_{\rm e}$.   To have a polar code with error probability less than $P_\text{e}$, how does the 
required block-length $N$ scale with the rate $R$?
\end{question}

Finding a suitable answer to the above questions is an easier task when the channel $W$ is a BEC. This is due to the simple closed form expression of
the process $Z_n$ given in \eqref{Z_n}.  
In the next section (Section~\ref{heur}), we provide heuristic methods that lead to 
suitable numerical answers to  Questions~\ref{Q1} and \ref{Q3} for the BEC. As we will see in the next section, such heuristic derivations are 
in excellent compliance with numerical experiments.
Using such  derivations, we also give an answer to Question~\ref{Q4} for the BEC.  

The heuristic results of Section~\ref{heur} provide us then with a concrete path to analytically tackle the above questions. In Section~\ref{scale_anal}
we provide analytical answers to  Questions~\ref{Q1}-\ref{Q4} for the BEC as well as other BMS channels. 
Providing a complete answer to Questions~\ref{Q1}-\ref{Q4} is beyond what we achieve in Section~\ref{scale_anal}, nevertheless, we provide close and useful  bounds.  Finally, in Section~\ref{conclusion} we conclude the paper.

\section{Heuristic Derivation for the BEC} \label{heur}
In this section we provide a heuristic (and numerical) procedure that leads to a clear picture of how the process 
$Z_n$ evolves through time $n$ when the channel $W$ is a BEC. 
As we will see, this procedure guides us to a number of conclusions 
about the process $Z_n$ which we refer to as \emph{assumptions}. By using these assumptions we can (numerically)  compute 
the important parameters for the process $Z_n$ which will then enable us to predict scaling laws for the evolution of 
$Z_n$ as well as scaling laws for polar codes. Several plots are provided to show the excellent compliance of these scaling predictions with  reality.  
The intuitive discussions as well as the numerical observations of this section will then 
help us in building a rigorous framework for the analysis of the evolution of $Z_n$. 
This is the subject of the next section (Section~\ref{scale_anal}). 
Let us emphasize that none of heuristic assumptions of the current section (Section~\ref{heur}) 
will be used in any of the proofs of the next section.   

Throughout this section we assume that the channel $W$ is the BEC($z$) where $z \in [0,1]$. 
To avoid  cumbersome notation, let
us define\footnote{${\rm Pr}(Z_n \in [a,b] \,  \big | \, Z_0=z  )$ and ${\rm Pr}(Z_n \in [a,b] )$ denote the same concept. We occasionally use the longer one only for the sake of a better illustration.}
\begin{equation} \label{p_n}
p_n(z,a,b)=  {\rm Pr}(Z_n \in [a,b] \,  \big | \, Z_0=z  ),
\end{equation}
where the condition $Z_0 = z$ means that $Z_n$ is the Bhattacharyya process of the  BEC($z$). We start by noticing that by \eqref{Z_n} the function $p_n(z,a,b)$ satisfies the following recursion
\begin{equation}\label{p_n}
 p_{n+1}(z,a,b)=\frac{p_n(z^2,a,b)+ p_n(1-(1-z)^2,a,b) }{2},
\end{equation}
with 
\begin{equation}\label{p-init} 
p_0(z,a,b)= \mathbbm{1}_{\{z \in [a,b] \}}.
\end{equation}
More generally, one can easily observe the following. Let $g: [0,1] \to \mathbb{R}$ 
be an arbitrary bounded function. Define the functions $\{g_n\}_{n \in \mathbb{N}}$, $g_n :[0,1] \to \mathbb{R}$, as
\begin{equation}\label{31}
g_n(z) = \mathbb{E} [ g(Z_n) \,  \big | \, Z_0=z].
\end{equation}
The functions $\{g_n\}_{n \in \mathbb{N}}$ satisfy the following recursion for $n \in \mathbb{N}$
\begin{equation}
g_{n+1}(z)= \frac{g_n(z^2) + g_n(1-(1-z)^2)}{2}.
\end{equation}
This observation motivates us to define the {\em polar operator}, denoted by $T$,  as follows. 
Let $\mathcal{B}$ be  the space of all bounded and real valued functions $g$ over $[0,1]$. The polar operator $T: \mathcal{B} \to \mathcal{B}$ maps a function $g \in \mathcal{B}$ 
to another function in $\mathcal{B}$ in the following way
\begin{equation}\label{pol_op}
T(g)= \frac{g(z^2) + g(1-(1-z)^2)}{2}.
\end{equation}   
It is now clear that  
\begin{equation} \label{g-T}
 \mathbb{E}[g(Z_n) \,  \big | \, Z_0=z ]=\overbrace{T \circ T \circ \cdots \circ T (g)}^{n \text{ times}} \triangleq T^n(g).
\end{equation}
In this new setting, our objective is to study the limiting behavior as well as the dynamics of the functions $T^n(g)$ when $g$ is a simple function  as in \eqref{p-init}.
This task is  intimately related to studying the eigenvalues of the polar operator $T$ and their corresponding eigenfunctions.
Also, a check shows that both of the functions 
\begin{align} \label{v_0,v_1}
v_0(z)=1, v_1(z)=z, 
\end{align}
are eigenfunctions associated to the eigenvalue  $\lambda=1$.

Consider now a function $g \in \mathcal{B}$. For simplicity, let us also assume that $g$ is continuous at $z =0$ and $z=1$. 
By using the fact that $Z_n$ polarizes, it is easy to see that
\begin{equation*}
\mathbb{E} [g(Z_n) \,  \big | \, Z_0=z] \stackrel{n \to \infty}{ \longrightarrow} (1-z)g(0) + z g(1).
\end{equation*}
Equivalently by \eqref{g-T} we have
\begin{equation} \label{T-conv}
T^n (g) \stackrel{n \to \infty}{ \longrightarrow} g(0) - z (g(1) - g(0)).
\end{equation}
In other words, $T^n(g)$ converges to a linear combination of the two eigenfunctions $v_0(z)=1$ and $v_1(z)=z$ that are associated to the eigenvalue  $\lambda=1$.
However, our main interest is to find out how fast the convergence in \eqref{T-conv} is taking place in terms of $n$. In this regard, to keep things simple and in a more manageable setting,  let us consider finite-dimensional approximations  of $T$. 
This is done by discretizing the unit interval into very small sub-intervals with the same length and by assuming that $T$ operates on all the points of 
each sub-interval in the same way. More concretely, consider a (large) number $L \in \mathbb{N} $ and let the 
numbers $x_i$, $i\in \{0,1,\cdots, L-1\}$, be defined as $x_i= \frac{i}{L-1}$. Hence, the unit interval $[0,1]$ can be thought of  as the union of the small 
sub-intervals $[x_i,x_{i+1}]$. Now, for simplicity assume that $g$ is a (piece-wise) continuous  function on $[0,1]$. 
Intuitively, by assuming $L$ to be large, we expect that the value of $g$ 
is the same throughout each of the intervals $[x_i,x_{i+1})$. Such an assumption seems also reasonable for the function $T(g)$ given in \eqref{pol_op}.
Thus, we can approximate the function $g$ as an $L$ dimensional vector 
\begin{equation} \label{g_L}
g_L \approx  [ g(x_0), g(x_1), \cdots , g(x_{L-1}) ].
\end{equation} 
In this way, from \eqref{pol_op} we expect that the function $T(g)$ can be well approximated by a matrix multiplication 
\begin{equation}
T (g)\approx g_L T_L, 
\end{equation}
where the $L \times L$ matrix $T_L$ is defined as follows. 
Let $T_L(i,j)$ be an element of $T_L$ in the $i$-th row and the $j$-th column. 
Define $T_L(1,1)=T_L(L,L)=1$ and for the other elements $i,j \in \{0,1, \cdots, L-1\}$ we let
\begin{equation} \label{T_L}
T_L(i+1,j+1)=  \left\{
\begin{array}{lr}
 \frac 12, &  \text{if $i =\lfloor (L-1) (\frac{j}{L-1})^2 \rfloor$},\\
 \\
\frac 12, &   \text{if $i=\lceil (L-1) (1- (1-\frac{j}{L-1})^2) \rceil$},\\
\\
0, & \text{o.w.}
\end{array} \right.
\end{equation}
As an example, the matrix $T_L$ for $L=10$ has the following form
\[ T_{10}= \left( \begin{array}{cccccccccc}
         1   &      \frac 12    &     \frac 12   &      0             &       0             &     0             &      0            &    0            &     0            &      0 \\
         0   &      0               &     0              &      \frac 12  &       \frac 12  &     0             &      0            &    0            &     0            &      0 \\
         0   &      \frac 12    &     0              &      0             &       0             &     \frac 12  &      0            &    0            &     0            &      0 \\
         0   &      0               &     0              &      0             &       0             &     0             &      0            &    0            &     0            &      0 \\
         0   &      0               &     \frac 12   &      0             &       0             &     0             &      \frac 12 &    0            &     0            &      0 \\
         0   &      0               &     0              &      \frac 12  &       0             &     0             &      0            &    \frac 12 &     0            &      0 \\
         0   &      0               &     0              &      0             &       0             &     0             &      0            &    0            &     0            &      0 \\
         0   &      0               &     0              &      0             &       \frac 12  &     0             &      0            &    0            &     \frac 12 &      0 \\
         0   &      0               &     0              &      0             &       0             &     \frac 12  &      \frac 12 &    0            &     0            &      0 \\
         0   &      0               &     0              &      0             &       0             &     0             &      0            &    \frac 12 &     \frac 12 &      1 \\
\end{array} 
\right).\]   
All the columns of $T_L$ sum up to $1$.  Hence, an application of 
the Perron-Frobenius  theorem \cite[Chapter 8]{PF} shows that the eigenvalues of $T_L$ are all inside the interval $[-1,+1]$. Also,  
a check shows that the matrix $T_L$ has an eigenvalue equal to $\lambda_0=1$ with two corresponding (left) eigenvectors 
\begin{align*}
&v_{0,L}= (1, 1, \cdots, 1),\\
&v_{1,L}= (x_0,x_1, \cdots, x_{L-1}),
\end{align*}
where $x_i = \frac{i}{L-1}$. By using \eqref{g_L}, it is easy to see that  the vectors $v_{0,L}$ and $v_{1,L}$ are the corresponding $L$-dimensional approximations of the eigenfunctions $v_0$ and $v_1$ given in \eqref{v_0,v_1}. We thus expect 
\begin{align} \label{g_L-T_L}
g_L T_L^n  \stackrel{n \to \infty}{\longrightarrow} c_0 v_{0,L} + c_1 v_{1,L},
\end{align}
where $c_0$ and $c_1$ are constants. Moreover, from \eqref{T-conv} we have
\begin{align*}
 & c_0 \stackrel{L \to \infty}{\longrightarrow} g(0), \\
 & c_1  \stackrel{L \to \infty}{\longrightarrow} g(1) - g(0).
\end{align*}
In order to find out how fast the convergence in \eqref{g_L-T_L} is, we  look at the 
second and third largest eigenvalues (in absolute value) of $T_L$ as $L$ grows large. 
We denote the second largest eigenvalue of $T_L$ by $\lambda_2(L)$, and the third largest eigenvalue is denoted by $\lambda_3(L)$. 
Table~\ref{eig-L} contains the value of these eigenvalues computed numerically for several (large) values of $L$.
\begin{table}
\centering
\begin{tabular}{c c c c c c c c c  }
\hline\hline 
 $L$ & $1000$ & & $2000$  & & $4000$ & & $8000$   \\
\hline
$\lambda_2(L)$  & $0.8227$ & & $0.8240$ & & $0.8248$ & & $0.8253$ \\
\hline
$\lambda_3(L)$    & $0.6878$ & &  $0.6958$ & &$0.7012$ & & $0.7046$   \\
\hline
\end{tabular}
\caption{{\small Values of $\lambda_2(L)$ and $\lambda_3(L)$, which correspond to the second and third largest eigenvalues of $T_L$ (in absolute value), are computed numerically for different values of $L$.}}
\label{eig-L}
\end{table}  
 It can thus be conjectured that
\begin{align}
\label{lambda_2} &\lim_{L \to \infty} \lambda_2(L) \approx 0.826, \\
\label{lambda_3} &\lim_{L \to \infty} \lambda_3(L) \approx 0.705.
\end{align}
This belief guides us to conclude that, for $L$ growing large, if we start from any vector $g_L$ then 
\begin{equation} \label{T_L_approx}
g_L T_L^n  \approx c_0 v_{0,L} + c_1 v_{1,L} + c_2 \lambda_2^n v_2 + O(n \lambda_3^n).  
\end{equation}
The above approximate relation indicates that for large $L$, the distance of $g_L T_L^n $ from its value in the limit is roughly equal to $c_2 \lambda_2^n$. 

One particular instance of the function $g$, is the one given in \eqref{p-init}, i.e., $g(z) = \mathbbm{1}_{\{z \in [a,b]\}}$. If $a,b \in (0,1)$ we know that $T^n(g) = \text{Pr} (Z_n \in [a,b])$ converges to $0$ everywhere (see \eqref{p_n}). If we consider the $L$-dimensional approximations of $g$ and $T$ for $L$ large, then the final limit of $g_L T_L^n$ would be arbitrarily close to $0$ (depending on how large $L$ is). Also, by \eqref{T_L_approx} the distance to this final limit is around $\lambda_2^n = 2^{ -n \log \frac{1}{\lambda_2}}$.  In words, the \textit{speed} of this convergence is $\log \frac{1}{\lambda_2}$.    
Now, let us go back the  original polar operator $T$ defined in \eqref{pol_op}. As we argued above, the operators $T_L$, for $L$ large, are good 
finite-dimensional approximations of $T$. The (experimental) relation \eqref{T_L_approx} brings us to the following assumption about $T$.

\begin{assumption}[Scaling Assumption]\label{SA}
There exists $\mu\in (0,\infty)$ such that, for any $z,a,b \in (0,1)$ such that $a<b$, the limit
$\lim_{n\to\infty} 2^{\frac{n}{\mu}} p_n(z,a,b)$ exists in $(0,\infty)$.
We denote this limit by $q(z,a,b)$. In other words,
\begin{equation} \label{SL}
\lim_{n\to\infty} 2^{\frac{n}{\mu}} {\rm Pr}(Z_n \in [a,b]) =q(z,a,b).
\end{equation}
We call the value $\mu$ the {\emph{scaling exponent}} of polar codes for the BEC.
\end{assumption}

By \eqref{SL} the value of ${\rm Pr}(Z_n \in [a,b])$ converges to $0$ like $2^{- \frac{n}{\mu}}$. Hence, the \textit{speed} of polarization for the process $Z_n$ over the BEC is equal to $\frac{1}{\mu}$.

Note here that by \eqref{lambda_2} we expect that  
\begin{equation}
2^{-\frac{1}{\mu}} = \lim_{L \to \infty} \lambda_2 (L)  \approx 0.826 \Rightarrow \frac{1}{\mu} \approx 0.275.
\end{equation}
Let us now describe a numerical method for computing $\mu$ and $q(a,b,z)$.
In this regard, we follow the approach of \cite{KMTU10}.  
First, by \eqref{p_n} and the scaling law assumption we conclude that
\begin{equation}\label{q(x)}
 2^{-\frac{1}{\mu}}q(z,a,b)=\frac{q(z^2,a,b)+ q(1-(1-z)^2,a,b)}{2}. 
\end{equation}
Equation \eqref{q(x)} can be solved numerically by recursion. In general, 
this equation can have many solutions. The idea here is to use the scaling assumption 
to properly initialize a recursion procedure to compute the desired solution of \eqref{q(x)} that is compatible with \eqref{SL} (i.e., a recursion that gives us the desired function $q$ in \eqref{SL}).
Let us now describe the recursion.
First of all, note that equation \eqref{q(x)} is invariant under multiplicative 
scaling of $q$. Also, from this equation one can naturally guess that $q(z,a,b)$
can be factorized into
\begin{equation} \label{p_invariant}
 q(z,a,b)=c(a,b) q(z),
\end{equation}
where $q(z)$ is a solution of \eqref{q(x)} with\footnote{Note that choosing $q(\frac 12)=1$ is an arbitrary normalization choice.}  $q(\frac 12)=1$.
We  iteratively compute $\mu$ and $q(z)$. 

Initialize $q_0(z)$ --say-- with\footnote{This is an arbitrary choice for $q_0(z)$. One can try other starting 
points, e.g., $q_0(z)= \mathbbm{1}_{\{ z \in [a,b] \}}$ or $q_0(z) = 4z(1-z)$. All the initial points that we have tried have 
led to the same $q(z)$. This is indeed compatible with the scaling assumption and \eqref{p_invariant}.}   
$q_0(z)=\mathbbm{1}_{\{z \in [ \frac 14, \frac 34] \}}$ and compute recursively new estimates of $q_{n+1}(z)$ by first computing
\begin{align*}
  \hat{q}_{n+1}(z) = & q_{n}(z^2) +  q_n(1-(1-z)^2),
\end{align*}
and then by normalizing $q_{n+1}(z)=  \hat{q}_{n+1}(z) /  \hat{q}_{n+1}(\frac12)$,
so that $q_{n+1}(\frac12)=1$. It is easy to see that $\hat{q}_{n}$ indeed converges to $q(z)$ 
provided that the scaling assumption as well as \eqref{p_invariant} hold true. 
We have implemented the above functional recursion numerically
by discretizing the $z$ axis. 
Figure~\ref{fig:scaling} shows the resulting numerical approximation of 
$q_{\infty}(z)$ as obtained by iterating the above procedure
until $\Vert q_{n+1}(z) - q_n(z)\Vert_{\infty}\leq10^{-10}$ ($\forall z \in [0,1]$) 
and by using a discretization with $10^6$ equi-spaced values of $z$.
\begin{figure}[ht] 
\begin{center} \input{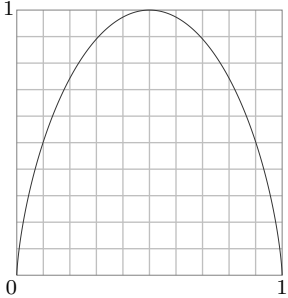} 
\end{center}
\caption{The function $ q(z)$ for 
$z \in [0, 1]$.}\label{fig:scaling}
\end{figure}
From this recursion we also get a numerical estimate of the scaling
exponent $\mu$. In particular we expect 
$ \hat{q}_n(1/2)\to 2^{1-\frac{1}{\mu}}$ 
as $n\to\infty$. Using this method, we 
obtain the estimate $1/\mu \approx 0.2757$.

\begin{figure}[h] 
\begin{center} \input{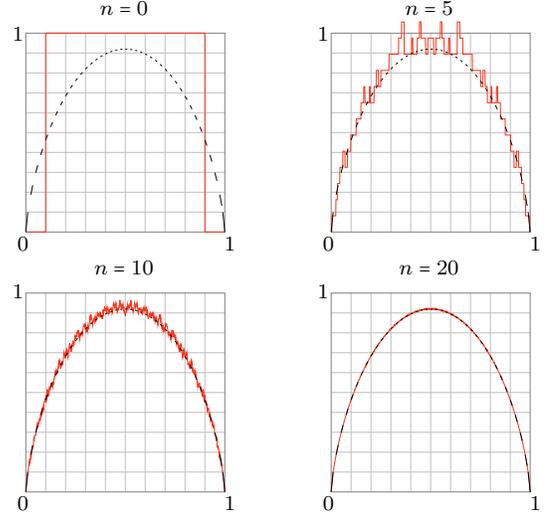} 
\end{center}
\caption{The functions $ 2^{\frac{n}{\mu}}q_n(a,b,z)$ for various values of 
$n$. Here we have fixed $a=1-b=0.9$ and $\frac{1}{\mu}=0.2757$.
In all of the four plots the dashed curve corresponds to
 $c(a,b) q(z)$ with $c(a,b)=0.92$.  Here, the function $q(z)$ corresponds to the numerical solution of \eqref{q(x)}. \label{fig:pemps}}
\end{figure}
As mentioned above, the function $q(a,b,z)$ differs from $q(z)$ by 
a multiplicative
constant $c(a,b)$ that is to be found by other means. In 
Figure~\ref{fig:pemps} we plot the functions 
$2^{\frac{n}{\mu}}q_n(z,a,b)$ for 
$a=1-b=\frac{1}{10}$ and different values of $n$. We observe that, as $n$ 
increases these plots and the curve 
$c(a,b) q(z)$ with $c(a,b)=0.92$ match very well. Even for moderate values of $n$ (such as $n=10$) we observe that  the curves have a fairly good agreement.

Let us now see what  the scaling law assumption implies about the finite-length behavior of polar codes.
For simplicity, we assume that communication takes place on the BEC($\frac 12$). We are given a target error 
probability $P_{\rm e}$ and want to achieve a rate of value at least $R$. 
What block-length $N$ should we choose?  

Consider the process $Z_n$ with $Z_0=z=\frac 12$. It is easy to see that the set of possible 
values that   $Z_n$ takes in $[0,1]$ is  symmetric  around $z=\frac 12$.  
Now, according to the scaling law for $x \in [0,\frac 12]$, there is a constant $q(\frac 12,x, \frac 12) \triangleq c(x)$ such that 
\begin{equation}
{\rm Pr}(Z_n \in [x, \frac 12]) \approx c(x) 2^{-\frac{n}{\mu}},
\end{equation}
As as result, by noticing the fact that $Z_n$ is symmetric around the point $z=\frac 12$, we get
\begin{equation} \label{ex:rel}
{\rm Pr}(Z_n \in [0, x]) \leq \frac 12- c(x) 2^{-\frac{n}{\mu}}.
\end{equation}   
From the construction procedure of polar codes (and specially relation \eqref{P_e}), we know the following. Let  $z(1) \leq z(2) \cdots \leq z(N)$ be a re-ordering
of the $N$ possible outputs of $Z_n$ in an ascending order. Then, by using \eqref{P_e} the error probability of a polar code with rate $R$ is bounded from below by\footnote{Note that if $W$ is a BEC, then we have $Z(W) = 2 E(W)$. Also, for general BMS channels we have the relation \eqref{EZ_bound}.} 
\begin{equation}
P_{{\rm e}} \geq \max_{i \in \mathcal{I}_{N,R} } E(W_N^{(i)}) = 
\max_{i \in \{1, \cdots, N \cdot R\} } \frac{z(i)}{2}  = \frac{z(N \cdot R)}{2} .
\end{equation}  
So in order to achieve error probability $P_{\rm e}$, we should certainly have $\frac{z(N \cdot R)}{2} \leq P_{\rm e}$ or $z(N \cdot R) \leq  2  P_{\rm e}$. 
As a result, we obtain
\begin{align*}
R  \leq  {\rm Pr}(Z_n \in [0, 2 P_{\rm e}]) ,
\end{align*}
and by using \eqref{ex:rel} we deduce that
\begin{align*}
R  & \leq  {\rm Pr}(Z_n \in [0, 2 P_{\rm e}]) \\
 & \leq \frac 12 -c( 2 P_{\rm e}) 2^{-\frac{n}{\mu}}\\
 &=\frac 12 - c(2  P_{\rm e}) N^{-\frac{1}{\mu}},
\end{align*} 
and finally
\begin{equation}
N \geq \biggl (\frac{c( 2 P_{\rm e})}{\frac 12 - R} \biggr )^{\mu}.
\end{equation} 
Now, from the above calculations we know that $\frac{1}{\mu}\approx 0.2757$. As a result, for the channel $W=\text{BEC}(\frac 12)$ we have
\begin{equation}
 N \geq \Theta \biggl(\frac{1}{(I(W)-R)^{3.627}}\biggr).
\end{equation}
For other empirical scaling laws of this type, we refer to \cite{KMTU10}.
In the next section, we provide methods that analytically validate the above observations. 
We also extend some of these observations and results to other BMS channels.

\section{Analytical Approach: from Bounds for the BEC to Universal Bounds for BMS Channels} \label{scale_anal}
In this section we provide a rigorous basis for the observations and derivations of the previous 
section. Proving the full picture of Section~\ref{heur} is beyond what we achieve here, but we come up with close and useful bounds.  
As previously mentioned, we only use 
the heuristic arguments as well as the numerical observations of the previous section to give an intuitive picture for 
the ideas and proofs of this section. 
In other words, the proofs of this section do not rely on any of the assumptions of the previous section and can be read independently.

This section consists of three smaller parts. In the first part we provide lower and upper bounds on 
 the speed of polarization for the BEC family. Similar types  of bounds are obtained for general BMS channels in the second part. Finally, in the last part we use these bounds to derive trade-offs between the
rate and the block-length for polar codes.   

\subsection{Speed of Polarization for the BEC Family} \label{mu-BEC}
The (heuristic) arguments of the previous section led us to the conclusion that  (see \eqref{SL}) for the channel $W=$BEC($z$) the quantity 
$\text{Pr}(Z_n \in [a,b])$ vanishes in $n$ like $\Theta(2^{-\frac{n}{\mu}})$ (here, $z,a,b \in (0,1)$).  
In other words, the speed of polarization for the process $Z_n$ is equal to $\frac{1}{\mu}$. The value of $\mu$ was also computed to be $\mu \approx 3.627$ (or $\frac{1}{\mu} \approx 0.2757$).

Analytically speaking, proving the scaling assumption \eqref{SL} seems to be a difficult task. It is not even clear whether the value $\mu$ exists.  
The objective of this section is to provide (analytical) lower and upper bounds on the value $\mu$. More precisely, 
we look for numbers $\underline{\mu}$ and $\overline{\mu}$ such that 
$\Theta(2^{-\frac{n}{\underline{\mu}}}) \leq \text{Pr}(Z_n \in [a,b] ) \leq \Theta(2^{-\frac{n}{\overline{\mu}}})$, 
or in words, the speed of polarization of $Z_n$ is bounded between the values $\frac{1}{\overline{\mu}}$ and $\frac{1}{\underline{\mu}}$.

In this regard, we provide two approaches that exploit different techniques.  
The  first approach is based on a more careful look at  equation \eqref{pol_op}. From the arguments of the previous section, the value $\mu$ is related to a significant (and non-trivial) eigenvalue of  the polar operator $T$. Here, we observe that  
simple bounds can be derived on the this eigenvalue of $T$ by carefully analyzing the effect of $T$ on some suitably chosen test functions.  
This approach provides  us with a sequence of 
analytic bounds on $\mu$. We conjecture (and observe empirically) that these bounds indeed converge to the value of $\mu$ that is computed in Section~\ref{heur}. 
The second approach considers all the possible compositions of the two operations $z^2$ and $2z-z^2$ and analyzes the asymptotic behavior of  these compositions. 
This approach provides us with 
a good lower bound on $\mu$. 
\subsubsection{First Approach} \label{second_approach}
Let us begin by providing an intuitive picture behind the first approach. This picture is only intended for a better explanation of the contents that appear later. Hence, these explanations can be skipped without losing the main track.   
Consider the polar operator defined in \eqref{pol_op} and its eigenvalues which are
the solutions of 
\begin{equation} \label{eigen}
T(q)= \lambda q. 
\end{equation}
A check shows that both $q(z)=z$ and $q(z)=1$ are eigenfunctions associated to the eigenvalue  $\lambda=1$.
Perhaps more interestingly, let us look at the eigenvalues of $T$ inside the interval $(0,1)$.   
Intuitively, equation \eqref{q(x)} together with the  scaling law \eqref{SL} can be reformulated as follows. 
The operator $T$ has an eigenvalue $\lambda \triangleq 2^{-\frac{1}{\mu}}$ and a corresponding eigenfunction $q(z)$ such that  if we take any step function 
$f(z)=\mathbbm{1}_{ \{z \in [a,b]\}}$, then
\begin{equation}\label{op_lim}
\lambda^{-n} T^n(f) \stackrel{n \to \infty}{\longrightarrow} c(a,b) q(z).
\end{equation}
Therefore, for $f(z)=\mathbbm{1}_{ \{z \in [a,b]\}}$, the value of  $T^n(f)$ vanishes in $n$ like $\Theta(\lambda^n)$ (or equivalently $\Theta(2^{-\frac{n}{\mu}})$). 
In fact, if the scaling law is true, then we naturally expect that \eqref{op_lim} holds for a much larger class of functions rather than the class of step functions. 
Heuristic arguments of the previous section also suggest that \eqref{op_lim} holds at least  for all (piece-wise) continuous  
functions $f(z)$ with $f(0)=f(1)=0$. Therefore, for any function $f$ in this larger class of functions the value of $T^n(f)$ decays like $\Theta(\lambda^n)$ (or equivalently $\Theta(2^{-\frac{n}{\mu}})$). So to compute (or to provide bounds on) the value of $\mu$, one can look for suitable continuous functions $f$  such that the speed of decay of $T^n(f)$ is ``easy" to compute (or provide bounds on). As we will see,  functions in the form of $f(z) = z^\alpha (1 - z)^\beta$ are among such suitable functions.  

Motivated by this picture, let us formalize the first approach to find bounds on the speed of polarization of $Z_n$ (or the value $\mu$) 
through of the following two steps: 
(1) choose a suitable ``test function" $f(z)$ for which we can provide good
bounds on how fast $T^n(f)$ approaches $0$ (in $n$), and (2) turn these bounds into bounds on the 
speed for polarization of $Z_n$ (or $\mu$).  
With this in mind, for a generic test function $f(z): [0,1]\to[0,1]$, let us define 
the sequence of functions $\{f_n(z) \}_{n \in \mathbb{N} }$ as $f_n:[0,1]\to[0,1]$ and for $z \in[0,1]$,
\begin{equation}
f_n(z) \triangleq \mathbb{E}[f(Z_n) \, \big | \, Z_0 = z] =T^n(f).
\end{equation}
Here, note that for $z \in [0,1]$ the value of $f_n(z)$ is a deterministic value 
that is dependent on the choice of $f$ and the process $Z_n$ with the starting value $Z_0=z$.  
Let us now recall once more the recursive relation of the functions $f_n$:
\begin{align} \label{f_n}
& f_0(z)=f(z), \\
& f_{n}(z)= \frac{f_{n-1}(z^2)+ f_{n-1}(1-(1-z)^2)}{2}. \nonumber
\end{align}
In order to find lower and upper bounds  on the speed of decay of the sequence $f_n$, we define  sequences of numbers $\{a_m\}_{m\in \mathbb{N}}$ and 
$\{b_m\}_{m\in \mathbb{N}}$   as
\begin{align}
& a_m = \inf_{z\in (0,1)} \frac{f_{m+1}(z)}{f_m(z)}, \label{a_m}\\
& b_m = \sup_{z\in (0,1)} \frac{f_{m+1}(z)}{f_m(z)}. \label{b_m}
\end{align}
\begin{lemma}\label{a_bec}
Fix $m\in \mathbb{N}$. For all $n \geq m$ and $z \in (0,1)$, we have
\begin{equation} \label{BEC_bound}
(a_m)^{n-m} f_{m} (z) \leq f_{n}(z) \leq (b_m)^{n-m} f_{m} (z)  .
\end{equation}
Furthermore, the sequence $a_m$ is an increasing sequence and the sequence $b_m$ is a decreasing sequence.
\end{lemma}
\begin{proof}
Here, we only prove the left-hand side of \eqref{BEC_bound} and note that the right-hand side follows similarly. 
The proof goes by induction on $n-m$. For $n-m=0$ the result is trivial. 
Assume that the relation \eqref{BEC_bound} holds for a  $n-m \triangleq k$, i.e., for $z\in (0,1)$ we have
\begin{equation} \label{BEC_m}
(a_m)^{k} f_{m} (z) \leq f_{m+k}(z).
\end{equation}
 We show that \eqref{BEC_bound} is then true for  $k+1$ and $z\in (0,1)$.  We have
 \begin{align*}
 f_{m+k+1}(z) & \stackrel{(a)}{=} \frac{f_{m+k}(z^2) + f_{m+k}(1-(1-z)^2)}{2} \\
 & \stackrel{(b)}{\geq} \frac{(a_m)^k f_{m}(z^2) + (a_m)^k f_{m}(1-(1-z)^2)}{2} \\
 & = (a_m)^k f_{m+1}(z)  \\
 & = (a_m)^{k} \frac{f_{m+1}(z)}{f_m(z)} f_m(z) \\
& \geq  (a_m)^{k}  \bigl [\inf_{z \in (0,1)}  \frac{f_{m+1}(z)}{f_m(z)} \bigr] f_m(z) \\
&= (a_m)^{k+1} f_m(z). 
 \end{align*}
 Here, (a) follows from \eqref{f_n} and (b) follows from \eqref{BEC_m},
 and hence the lemma is proved via induction.
 
 Finally, the sequence $a_m$ increases by $m$ because if we plug in
$k = 1$ to the above set of ineqqualities, 
and stop after the third line, then we obtain that $f_{m+2}(z) \geq a_m f_{m+1}(z)$ for $z\in (0,1)$. 
From this and the definition of $a_m$ in \eqref{a_m}, it is then easy to see that $a_{m+1} \geq a_m$. 
\end{proof}
 Let us now begin searching for suitable test functions, i.e.,  candidates for $f(z)$  that provide us with good lower and upper bounds $a_m$ and $b_m$. 
 First of all, it is easy to see that a test function $f(z)=\mathbbm{1}_{\{z \in [a,b]\}}$ results in trivial values of $a_m$ and $b_m$ (namely $b_m = \infty$ and $a_m$ is not well-defined), and hence such test functions are not suitable for this bounding technique.  
 Second, we expect that having a polynomial test function might be slightly preferable. This is due to the fact that if   
$f$ is a polynomial, then $T^n(f)$ is also  a polynomial and computing $a_m$ and $b_m$ is equivalent to finding roots of polynomials which is a manageable task. Of course the simplest polynomial that takes the value $0$ on $z=0,1$ is $f(z)=z(1-z)$. Hence, let us take our test function as 
$f(z)=z(1-z)$ and consider 
 the corresponding sequence of functions $\{f_n(z) \}_{n \in \mathbb{N} }$ with $f_0(z) \triangleq f(z) = z(1-z)$ and
\begin{equation}\label{fnz(1-z)}
f_n(z)=  \mathbb{E}[Z_n(1-Z_n)] =T^n(f_0).
\end{equation}
A moment of thought shows that with $f_0=z(1-z)$ the function $2^n f_n$ is a 
polynomial of degree $2^{n+1}$ with integer coefficients. Let us first focus on computing the value of $a_m$ for $m\in \mathbb{N}$. 
\begin{remark} \label{rem_a}
One can compute the value of $a_m$ by finding the extreme points
of the function $\frac{f_{m+1}}{f_m}$ (i.e., finding the roots of
the polynomial $g_m={f'}_{m+1}f_m-f_{m+1}{f'}_m$),
 and then minimizing the function $\frac{f_{m+1}}{f_m}$ on these extreme points as well as 
 boundary points\footnote{Note that in spite of the fact that the supremum and the infimum are defined for $z \in (0,1)$, we should check the value of $\frac{f_{m+1}}{f_m}$ around the boundary points $z = 0,1$.} $z=0,1$. Assuming $f_0=z(1-z)$, for small values e.g., $m=0,
1$, pen and paper suffice to find the extreme points. For higher values of $m$, we can
automatize the process: all these polynomials have rational
coefficients and therefore it is possible to determine the number
of real roots exactly and to determine their value to any desired
precision.
This task can be accomplished precisely by computing so-called Sturm chains
(see Sturm's Theorem \cite{sturm}). Computing Sturm chains is equivalent to
running Euclid's algorithm starting with the second and third
derivative of the original polynomial. Hence,
we can  \textit{analytically} find the value of $a_m$ to any desired precision.
Table~\ref{a_val} contains the numerical value of $a_m$ up to
precision $10^{-4}$ for $m\leq 10$. As the table shows, the values $a_m$ are 
increasing (see Lemma~\ref{a_bec}), and we conjecture that they converge to $2^{-0.2757}=0.8260$, the corresponding value for the 
channel BEC. 
\end{remark}
\begin{table}
\centering
\begin{tabular}{c c c c c c c c }
$m$ & $0$  & $2$  & $4$    & $6$  & $10$  \\
\hline
$a_m$ & $0.75$  & $0.7897$  & $0.8074$  & $0.8190$  & $0.8239$ \\
\hline
$\log a_m$ &$-0.4150$  & $-0.3406$  & $-0.3086$  & $-0.2880$ & $-0.2794$ \\
\end{tabular}
\caption{The values of $a_m$ corresponding to the test function $f_0=z(1-z)$ are numerically computed for several choices of $m$. }
\label{a_val}
\end{table}  

We now focus on  computing the value of $b_m$. 
On the negative side, for the specific test function $f(z)=z(1-z)$  
we obtain $b_m=1$ for $m\in \mathbb{N}$ and therefore the upper bounds 
implied by \eqref{b_m} are trivial. In fact, it is not hard\footnote{This follows from repeated applications of L'H\^{o}pital's rule.} to show that if we plug in any polynomial as the test function then we get $b_m=1$ for any $m$.
On the positive side, we can consider other test functions that result in non-trivial values for $b_m$. The problem with non-polynomial functions is that methods such as the Sturm-chain method no longer apply. Hence, finding the precise value of $b_m$ up to any desired precision can in general be a difficult task
and we might  lose the analytical tractability of $b_m$.
As an example,  choose 
\begin{equation}
f_0(z)=z^\alpha (1-z)^\beta,
\end{equation}
for some choice of $\alpha, \beta \in (0,1)$. Then, from \eqref{b_m} we have
\begin{equation} \label{b_0-alpha-beta}
b_0=\sup_{z\in (0,1)} \frac{f_1(z)}{f_0(z)}= \sup_{z \in [0,1]} \frac{z^\alpha (1+z)^\beta + (2-z)^\alpha (1-z)^\beta}{2}.
\end{equation}
We can compute $b_0$ to any desired precision either by 
finding the extreme points of the expression in \eqref{b_0-alpha-beta}, or by simple numerical methods.
\begin{remark}
Let us explain what we mean by a simple numerical method. The idea is to take a fine
grid for the unit interval and maximise the right-hand side of \eqref{b_0-alpha-beta} on this grid. Let $g(z) = \frac{z^\alpha (1+z)^\beta + (2-z)^\alpha (1-z)^\beta}{2}$. We now describe briefly a numerical procedure to find precisely the maximum value that $g$ attains over $[0,1]$: (i) Fix a number $\delta>0$. The function $g$ has a finite derivative on the interval $(\delta, 1-\delta)$. Thus, the maximum of $g$ over the interval $(\delta, 1-\delta)$ can be found to any desired precision by making the grid sufficiently fine. (ii) The maximum value of $g$ over the region $[0,\delta] \cup [1-\delta,1]$ can be upper-bounded by simple Taylor-type methods. This upper bound becomes tighter when $\delta$ is smaller. It is then straight-forward to conclude that by this procedure  we can compute, to any desired precision,  the maximum value that $g$ attains over the unit interval (provided that we choose a sufficiently small $\delta$ and grid size). \label{rem5}
\end{remark}
By letting $\alpha=\beta=\frac 23$, we obtain $b_0=0.8312$ which is already a good bound for $\lambda$ (recall from the calculations done in Section~\ref{heur} that $\lambda \approx 2^{-0.2757} = 0.8260$). 
This suggests that the test function $f_0(z)=(z(1-z))^{\frac 23}$ is a suitable candidate for obtaining good upper bounds $b_m$. 
For this specific test function, the value of $b_m$ for various values of $m$ has been numerically computed in Table~\ref{nb_n}. As we observe 
from Table~\ref{nb_n}, even for moderate values of $m$ the (numerically computed) bound $b_m$ is very close to
the ``true" value of $\lambda$.   
\begin{table}
 \centering
 \begin{tabular}{c c c c c c c c }
 $m$ & $0$ & $2$   & $4$   & $6$ & $8$  \\
 \hline
 $b_m$  & $0.8312$ & $0.8294$ & $0.8279$ & $0.8268$ &$0.8264$ \\
 \hline
 $\log b_m$  &$-0.2663$ & $-0.2699$ & $-0.2725$ & $-0.2744$ & $-0.2751$ \\
 \end{tabular}
 \caption{ The values of $b_m$ corresponding to $f_0=(z(1-z))^{\frac 23}$ are numerically computed for several choices of $m$. 
 \label{nb_n}}
 \end{table}


Finally, let us relate the bounds $a_m$ and $b_m$ to bounds on the value of $\text{Pr}(Z_n \in [a,b])$. 
This is the subject of the following lemma which is proven in Appendix~\ref{proofs}.  
\begin{lemma}\label{lim:equal}
Let $a,b \in (0,1)$ be such that $\sqrt{a}\leq 1-\sqrt{1-b}$. Then, there exists a constant $c_1 >0$  such that for any $z \in [0,1]$
\begin{multline} \label{32}
\frac{1}{n} \log \mathbb{E}[Z_n(1-Z_n)] - \frac{c_1 \log n }{n} \\ \leq \frac{1}{n}\log( 2^{-n} + {\rm Pr} (Z_n \in [a,b]) ) 
\end{multline}
where $Z_n$ is defined in \eqref{Z_n} with $Z_0 = z$.
Also, for any continuous function $f:[0,1] \to [0,1]$ such that $f(z) >0$ for $z \in (0,1)$, we have for $a,b \in (0,1)$ that
\begin{equation}\label{33}
\frac{1}{n}\log {\rm Pr} (Z_n \in [a,b]) \leq \frac{1}{n} \log \mathbb{E}[f(Z_n)] + \frac{c_3}{n},
\end{equation}
where $c_3$ is a positive constant that depends on $a,b,$ and $f$. Examples of such function $f$ can be $f(z) = z(1-z)$ or $f(z) = (z(1-z))^{\frac 23}$.
\end{lemma}
We can now easily conclude the following.
\begin{corollary}
Fix $m \in \mathbb{N}$. For $a,b \in (0,1)$ such that $\sqrt{a}\leq 1-\sqrt{1-b}$ and $n \geq m$ we have 
\begin{multline} \label{ampbm}
\log a_m + O(\frac{\log n}{n}) \leq \frac{1}{n} \log ( 2^{-n} +  {\rm Pr}(Z_n \in [a,b]) )  \\ \leq \log b_m + O(\frac{1}{n}),
\end{multline}
where $a_m$ is defined in \eqref{a_m} with the test function $f(z)=z(1-z)$ (see Table~\ref{a_val}), and $b_m$ 
is defined in \eqref{b_n} 
with the test function $f(z)=(z(1-z))^{\frac 23}$ (see Table~\ref{nb_n}).
\end{corollary} 
\begin{remark}
Two comments are in order: (i) The additional term $2^{-n}$ in \eqref{32} and \eqref{ampbm} is to avoid trivial conflicts when ${\rm Pr}(Z_n \in [a,b]) = 0$. However, these cases are very rare as for every  $z \in (0,1)$ and $a,b \in(0,1)$ s.t.  $\sqrt{a}\leq 1-\sqrt{1-b}$, it is not hard to prove that there exists an integer $n_0 \in \naturals$ such that we have ${\rm Pr} (Z_n \in [a,b]) >0$ for $n \geq n_0$. Note that if ${\rm Pr} (Z_n \in [a,b]) >0$, then we certainly have ${\rm Pr} (Z_n \in [a,b]) >2^{-n}$.
(ii) We expect that the result of of Lemma~\ref{lim:equal} holds
for any choice of $a$ and $b$ such that $a<b$. That is, the condition $\sqrt{a}\leq 1-\sqrt{1-b}$ is
only a technical restriction.
\end{remark}
\subsubsection{Second Approach} \label{second-approach}
 We will now explain an other approach for finding the value $\mu$ for the BEC.  Let us point out the fact that the content of this section (Section~\ref{second-approach}) is not necessary for the forthcoming parts of the paper and hence can be skipped  without losing the main track. 

Throughout this section we will prove the following theorem. 
\begin{theorem}\label{int_er}
We have  
\begin{equation}\label{low:ln}
\liminf_{n \to \infty} \frac{1}{n}\log \bigl \{ \int_0 ^1  {\rm Pr} (Z_n \in [a,b] ) dz \bigr \}  \geq  \frac {1}{2 \ln 2} - 1 \approx -0.2787.
\end{equation}
\end{theorem}
Let us now explain, at the intuitive level, the main consequence of Theorem~\ref{int_er}.  By using the scaling law assumption, and specifically \eqref{T_L_approx}  and  \eqref{SL}, we have that
$ \int_0^1 {\rm Pr} (Z_n \in [a,b] ) dz \approx \int_0^1  2^{-\frac{n}{\mu}}  q(z,a,b) dz+  o(2^{-\frac{n}{\mu}})$.
This relation together with \eqref{low:ln} implies that $\mu \geq \frac {1}{2 \ln 2} - 1 \approx -0.2787 $.    
For the sake of brevity, we do not address here further (analytic) conclusions of Theorem~\ref{int_er} and we refer the reader to \cite{HKU10}. 

To proceed with the proof of Theorem~\ref{int_er}, let us recall  from Section~\ref{polproc} the definition of $Z_n$ (for the BEC) in terms of the sequence $\{B_n\}_{n \in \mathbb{N}}$.
We start by $Z_0=z$ and
\begin{equation} \label{process_B}
Z_{n+1}=  \left\{
\begin{array}{lr}
Z_{n-1}^2 &  ; \text{if } B_n=1,\\
2Z_{n-1}-Z_{n-1}^2 &   ; \text{if } B_n=0.
\end{array} \right.
\end{equation}
Hence, by considering the two maps $t_0,t_1:[0,1]\longrightarrow [0,1]$ defined as
\begin{equation} \label{T}
t_0(z)=2z-z^2 , t_1(z)=z^2,
\end{equation}
the value of $Z_{n}$ is obtained by applying $t_{B_n}$ on the value of $Z_{n-1}$, i.e., 
\begin{equation}
Z_n = t_{B_n}(Z_{n-1}).
\end{equation}
 The same rule applies for obtaining the value of $Z_{n-1}$ form $Z_{n-2}$ and so on. Thinking this through recursively, the value of $Z_n$ is
 obtained from the starting point of the process,  $Z_0=z$, via the following (random) maps.\footnote{The necessary notation is reviewed in Section~\ref{polproc}. }
\begin{definition}\label{phi}
For each $n \in \mathbb{N}$ and a realization $(b_1, \cdots,b_n) \triangleq \omega_n \in \Omega_n$ define the map $\phi_{\omega_n}$ by
\begin{equation}
\phi_{\omega_n}=t_{b_n} \circ t_{b_{n-1}} \circ \cdots t_{b_1}.
\end{equation}
Also, let $\Phi_n$ be the set of all such $n$-step maps. 
\end{definition}
As a result, an equivalent description of the process $Z_n$ is as follows. At time $n$ the value of $Z_n$ is obtained by picking uniformly
 at random one of the functions  $\phi_{\omega_n} \in \Phi_n$ and assigning the value  $ \phi_{\omega_n}(z)$ to $Z_n$. Consequently we have,
\begin{align}\label{equ:equivalent}
{\rm Pr}(Z_n \in [a,b]) &= \sum_{\omega_n \in \Omega_n}\frac{1}{2^n} \mathbbm{1}_{\{\phi_{\omega_n}(z) \in [a,b]\}}.
\end{align}

By using \eqref{equ:equivalent}, it is apparent that in order to analyze the behavior of the quantity $ {\rm Pr}(Z_n \in [a,b]) $ as $n$ grows large, it is necessary to characterize the asymptotic behavior of the random maps $\phi_{\omega_n}$. Continuing
the theme of Definition~\ref{phi}, we can assign to each realization of the infinite sequence $\{B_k\}_{k \in \mathbb{N}}$, denoted by $\{b_n\}_{n \in \mathbb{N}}$, a
sequence of maps $\phi_{\omega_1}(z), \phi_{\omega_2}(z),\cdots$ where $\omega_i\triangleq (b_1, \cdots,b_i)$. We call the sequence
$\{\phi_{\omega_k}\}_{k \in \mathbb{N}}$ the corresponding sequence of maps for
the realization $\{b_k\}_{k \in \mathbb{N}}$. We also use the
realization $\{b_k\}_{k \in \mathbb{N}}$ and its corresponding
$\{\phi_{\omega_k}\}_{k \in \mathbb{N}}$ interchangeably.
Let us now focus on the asymptotic characteristics of the functions $\phi_{\omega_n}$. Firstly, since $\{\phi_{\omega_n } (z) \}_{\omega_n \in \Omega_n}$  has the same law as $Z_n$ starting at
$Z_0=z$, we conclude that for $z \in [0,1]$, with
probability one, the quantity $\lim_{k \to \infty} \phi_{\omega_k} (z)$ takes on a value in the set $\{0,1\}$ . In Figure \ref{sample_threshold} the  functions $\phi_{\omega_n}$ are 
plotted for a random realization. As it is apparent from the figure, the 
functions $\phi_{\omega_n}$ seem to converge point-wise to a jump function (i.e., a sharp rise from $0$ to $1$). An intuitive justification of this fact is as follows. 
Consider a random function $\phi_{\omega_n}$. Due to polarization, as $n$ grows large, 
almost all the values that this function takes are very close to $0$ or $1$. This function is also increasing and continuous (more precisely, it is a polynomial).
A little thought reveals that the only choice to imagine for $\phi_{\omega_n}$ is a very sharp rise from being almost $0$ to almost $1$.
The formal and complete statement is given as follows.
\begin{figure}[htb] 
\begin{center} \input{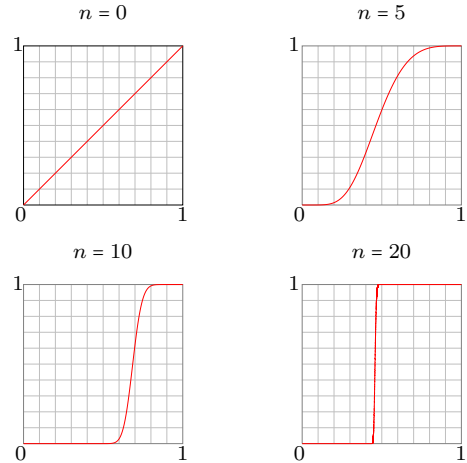} 
\end{center}
\caption{The functions $\phi_{\omega_n}$ associated to a random 
realization are plotted. As we see as $n$ grows large, the functions $\phi_{\omega_n}$ converge point-wise to a step function.}
 \label{sample_threshold}
\end{figure}

\begin{lemma}[Almost every realization has a threshold point]\label{threshold}
For almost every realization of $\omega \triangleq \{b_k\}_{k \in \mathbb{N}} \in \Omega$, there exists a point
$z_{\omega}^* \in [0,1]$, such that
\begin{equation}\nonumber
\lim_{n \to \infty} \phi_{\omega_n}(z) \rightarrow  \left\{
\begin{array}{lr}
0  & z \in [0,z_{\omega}^*) \\
1 &  z \in (z_{\omega}^*,1]
\end{array} \right.
\end{equation}
Furthermore, $z_{\omega}^*$ has uniform distribution on $[0,1]$. We call the point $z_{\omega}^*$ the threshold point of the realization $\{b_k\}_{k \in \mathbb{N}}$ or the threshold point of its corresponding sequence of maps $\{\phi_{\omega_k}\}_{k \in \mathbb{N}}$.
\end{lemma}
Looking more closely at \eqref{equ:equivalent}, by the above lemma we conclude that as $n$ grows large, the maps $\phi_{\omega_n}$ that activate the identity function $\mathbbm{1}_{\{\cdot\}}$ must have their threshold point sufficiently close to $z$. Let us now give an intuitive discussion about the idea behind the proof of Theorem~\ref{int_er}. By using \eqref{equ:equivalent}  we can write
\begin{align} \nonumber
 {\rm Pr}(Z_n \in [a,b])  & = \sum_{\omega_n \in \Omega_n}\frac{1}{2^n} \mathbbm{1}_{\{ \phi_{\omega_n}(z) \in  [a,b]\}}\\ \label{inv}
& = \sum_{\omega_n \in \Omega_n}\frac{1}{2^n} \mathbbm{1}_{\{z \in [\phi_{\omega_n}^{-1} (a),\phi_{\omega_n}^{-1}(b)]\}}. 
\end{align}
Hence, by Lemma~\ref{threshold}, for a large choice of $n$ the intervals $ [\phi_{\omega_n}^{-1} (a),\phi_{\omega_n}^{-1}(b)]$ have a very 
short length and are distributed almost uniformly along $[0,1]$. Now, if we assume that the length of the intervals $ [\phi_{\omega_n}^{-1} (a),\phi_{\omega_n}^{-1}(b)]$ 
is very close to their average, then we can replace the average in \eqref{inv} by the 
average length of $[\phi_{\omega_n}^{-1}(a) , \phi_{\omega_n}^{-1} (b)]$. 
That is, 
$${\rm Pr}(Z_n \in [a,b]) \approx \mathbb{E} [\phi_{\omega_n}^{-1} (b) -\phi_{\omega_n}^{-1}(a)]. $$
So intuitively, all that remains is to compute the average length of the random intervals $[\phi_{\omega_n}^{-1}(a) , \phi_{\omega_n}^{-1} (b)]$.

In fact we are not able to make all these heuristics  precise for the point-wise values $\frac{1}{n} \log {\rm Pr}(Z_n \in [a,b]) $. 
Nonetheless, the picture is naturally precise for the 
 average of  ${\rm Pr}(Z_n \in [a,b])$ over $z \in [0,1]$, i.e.,  
\begin{equation}\label{b_n}
\frac{1}{n} \log \bigl \{ \int_{0}^{1} {\rm Pr} (Z_n \in [a,b]) dz \bigr \}.
\end{equation}
To see this, we proceed as follows. By \eqref{inv} we have
\begin{align*}
 \int_{0} ^{1} {\rm Pr} (Z_n \in [a,b] ) dz &=  \int_{0} ^{1} \bigl \{ \sum_{\omega_n \in \Omega_n} \frac{1}{2^n} \mathbbm{1}_{\{z \in  \phi_{\omega_n} ^{-1} [a,b]\}} \bigr \} dz \\
&=  \sum_{\omega_n \in \Omega_n} \frac{1}{2^n} \bigl \{  \int_{0} ^{1}  \mathbbm{1}_{\{z \in  \phi_{\omega_n} ^{-1} [a,b]\}} dz \bigr \}\\
&= \mathbb{E} [\phi_{\omega_n}^{-1} (b) -\phi_{\omega_n}^{-1}(a)],
\end{align*}
and by applying $\frac{1}{n}\log(\cdot)$ to both sides we have
\begin{eqnarray} \nonumber
\frac{1}{n}\log \bigl \{ \int_0 ^1  {\rm Pr} (Z_n \in [a,b] ) dz \bigr \} \!\!\!\!\! &\!\!\!\!\!  = \!\!\!\!\! 
& \frac{1}{n} \log \mathbb{E} [\phi_{\omega_n}^{-1} (b) -\phi_{\omega_n}^{-1}(a))] \nonumber \\ 
 & \geq &  \frac{1}{n} \mathbb{E}[ \log (\phi_{\omega_n}^{-1} (b) -\phi_{\omega_n}^{-1}(a))],\nonumber \\
 &\label{jensen}
\end{eqnarray}
where in the last step we have used  Jensen's inequality. The value of $\lim_{n \to \infty}  \frac{1}{n} \mathbb{E}[
\log( \phi_{\omega_n}^{-1} (b) -\phi_{\omega_n}^{-1}(a) )] $ can be computed precisely.
\begin{lemma} \label{log_measure}
 We have
\begin{equation*}
 \lim_{n \rightarrow \infty}\frac{1}{n} \mathbb{E}  [\log (\phi_{\omega_n}^{-1} (b) -\phi_{\omega_n}^{-1}(a))]  = \frac {1}{2 \ln 2} - 1 \approx -0.2787.
\end{equation*}
\end{lemma}\
As a result, we have  
\begin{equation*}
\liminf_{n \to \infty} \frac{1}{n}\log \bigl \{ \int_0 ^1  {\rm Pr} (Z_n \in [a,b] ) dz \bigr \}  \geq  \frac {1}{2 \ln 2} - 1.
\end{equation*}

The result of Theorem~\ref{int_er}  provides a lower bound on $\mu$ that is very close to the value we obtained in Section~\ref{heur} but 
is not exactly equal.  This is because we have  used Jensen's inequality  in \eqref{jensen}.

\subsection{Speed of Polarization for General BMS Channels}
In the previous part, we derived bounds on the speed of polarization for the process $Z_n$ associated to the BEC.
To this end, we used the recursion \eqref{Z_n}  for $Z_n$ and the fact that the speed of 
polarization can be ``measured" by computing the rate of decay of a sequence $\{\mathbb{E}[f(Z_n)]\}_{n \in \naturals}$, where $f$ is a suitable ``test" function such as $f(z) = z(1-z)$ or $f(z) = (z(1-z))^{\frac 23}$.

In this part, we use a similar approach to bound the speed of polarization for \textit{any} BMS channel.   
For a BMS channel $W$, there is no simple and closed-form (scalar) recursion for the  process $Z_n$ as for the BEC. 
However, by using  \eqref{Zgen-} and \eqref{Zgen+},  we can provide bounds on how $Z_n$ evolves:
 \begin{equation} \label{Z_nextreme}
Z_{n+1}  \left\{
\begin{array}{lr}
={Z_n}^2 &  ; \text{if $B_n=1$},\\
\in [Z_n \sqrt{2-{Z_n}^2} , 2Z_n-{Z_n}^2] &  ; \text{if $B_n=0$}.
\end{array} \right.
\end{equation}
As a warm-up, we notice that similar techniques as used in Section~\ref{second_approach} can be used to provide  general lower and upper bounds.
For instance, to find upper bounds we can proceed as follows.  For any continuous function 
$g:[0,1]\to\mathbb{R}$ such that $g(0)=g(1)=0$ and $g(z)>0$ for $z \in (0,1)$, let 
\begin{align} \label{L_g}
L_g = \sup_{z\in(0,1),  y \in [z\sqrt{2-z^2}, z(2-z)] \}} \frac{g(z^2) +g(y)}{2g(z)}.
\end{align}
Similar to the discussion in Section~\ref{second_approach} (in particular the proof of Lemma~\ref{a_bec}), it is easy to see from  \eqref{Z_nextreme} and \eqref{L_g} 
that for $n \in \naturals$
\begin{align*}
\mathbb{E}[g(Z_{n+1}) | Z_n] \leq g(Z_n) L_g,
\end{align*} 
and consequently,
\begin{align*}
\mathbb{E}[g(Z_n)] \leq g(Z_0) L_g^n.
\end{align*}
As a result, for the process $Z_n=Z(W_n)$ we have
\begin{equation} \label{g-L_g}
 \mathbb{E}[g(Z_n)] \leq c L_g^n,
\end{equation}
where $c=\sup_{z \in [0,1]} g(z)$ is a constant. Also, by using the Markov inequality we have for
 $a,b \in (0,1)$,
\begin{equation} \label{log-L_g}
 \frac{1}{n} \log {\rm Pr} (Z_n \in [a, b]) \leq \log L_g + O(\frac{1}{n}) . 
\end{equation}
It thus remains to find good candidates for the function $g$ (with the properties mentioned above) such that the value $L_g$ defined in \eqref{L_g} is minimized. 
For instance, we can let the function $g$ take the following closed form: $g(z) =  (a z^2 + bz +c ) (z(1-z))^d$ where $a,b,c,d \in (0,1)$ and optimize the value of $L_g$ over the choice of $a,b,c,d$.  
For example, by choosing  $a=b=0$, $c=1$, and $d = \frac 23$  we have $g(z)=(z(1-z))^{\frac 23}$ and we obtain $\log L_g=-0.169$. That is
\begin{equation} \label{univ-up0}
  \mathbb{E}[(Z_n (1- Z_n))^{\frac 23}] \leq 2^{-0.169 n} ,
\end{equation}
Also, by choosing $ a= \frac 25$, $b=\frac 14$,  $c=\frac{19}{20}$, and $d = \frac 34$ we have $g(z) =  \frac{1}{20}(  8 z^2 + 5  z + 19  ) (z(1-z))^{\frac 34}$. The value of $L_g$ can be computed to a desirable precision using simple numerical methods (see Remark~\ref{rem5}).    We thus obtain $\log L_g = -0.202$ and as a result   
\begin{align} \nonumber
\mathbb{E} [g(Z_n)]  &=  \mathbb{E} [\frac{1}{20}(  8 Z_n^2 + 5  Z_n + 19  ) (Z_n (1- Z_n ))^{\frac 34}] \\ \nonumber
& \!\! \stackrel{\eqref{g-L_g}}{\leq} \!\! \bigl (\max_{z \in [0,1]} \frac{1}{20}(  8 z^2 + 5  z + 19  ) (z(1-z))^{\frac 34} \bigr ) \times 2^{-0.202 n} \\
& \leq \frac 12 2^{-0.202 n}. \label{g-L_g-1}
\end{align}
Also, by \eqref{log-L_g} we obtain
\begin{equation}
\frac{1}{n} \log {\rm Pr} (Z_n \in [a, b]) \leq -0.202 + O(\frac{1}{n}). 
\end{equation}  
As a final remark, we note that for $z \in[0,1]$ we have $g(z) \geq \frac 12 (z(1-z))^{\frac 34} $. Therefore, we can conclude that for any BMS channel $W$ we have
\begin{equation} \label{univ-up}
\mathbb{E} [(Z_n (1- Z_n ))^{\frac 34}] \leq 2  \mathbb{E}[g(Z_n)] \stackrel{\eqref{g-L_g-1}}{ \leq}  2^{-0.202 n}.
\end{equation}


The relations of type \eqref{univ-up0} and  \eqref{univ-up} are upper bounds on the speed of polarization that hold \emph{universally} over all  
BMS channels. 
Let us now compute universal lower bounds. In the rest of this section, it is more convenient for us to consider another stochastic process related to $W_n$, which is the process\footnote{For the BEC the processes $H_n$ and $Z_n$ are identical.} $H_n=H(W_n)$. The main reason to consider $H_n$ rather than $Z_n$ is that the process $H_n$ is a martingale and this martingale property will help us to use the functions $\{f_n\}_{n\in \mathbb{N}}$  defined in \eqref{f_n} (with the starting function $f(z)=z(1-z)$) to provide universal lower
 bounds on the quantity $\mathbb{E}[H_n(1-H_n)]$. We begin by introducing one further technical condition given as follows.
\begin{definition}
 We call an integer $m \in \mathbb{N}$ \emph{suitable} if the function $f_m(z)$, defined in \eqref{f_n} (with the starting function $f(z)=z(1-z)$), is concave on $[0,1]$.
\end{definition}
\begin{remark} \label{rem_b}
For small values of $m$, i.e., $m\leq2$, it is easy to verify by hand that the function $f_m$ is concave.
As discussed previously, for larger values of $m$ we can use Sturm's theorem \cite{sturm}
and a computer algebra system to verify this. Note that
the polynomials $2^mf_m$ have integer coefficients. Hence, all
the required computations can be done exactly. We have checked up to $m=10$
that $f_m$ is concave and we conjecture
 that in fact this is true
for all $m\in \mathbb{N}$.
\end{remark} 
We now show that  for any BMS channel $W$, the value of $a_m$, defined in \eqref{a_m},  is a 
lower bound on the speed of polarization of $H_n$ provided that $m$ is a suitable integer. 
\begin{lemma} \label{BMS}
Let $m \in \mathbb{N}$ be a suitable integer and $W$ a BMS channel with $I(W) \in (0,1)$. We have for $n \geq m$ 
\begin{equation}
 \mathbb{E}[H_{n}(1-H_{n})] \geq (a_m)^{n-m} f_m(H(W)),
\end{equation}
where $a_m$ is given in \eqref{a_m}.
\end{lemma}
\begin{proof}
We use induction on $n-m$: for $n-m=0$ there is nothing to prove. Assume that the result of the lemma
 is correct for $n-m=k$. 
Hence, for any BMS channel $W$ with $H_n=H(W_n)$ we have
\begin{equation} \label{BMS_m}
 \mathbb{E}[H_{m+k}(1-H_{m+k})] \geq (a_m)^k f_m(H(W)).
\end{equation}
We now prove the lemma for $m-n=k+1$. For the BMS channel $W$, let us recall from Section~\ref{chantrans} that the  transform
$W \to (W^0,W^1)$ yields two channels $W^0$ and $W^1$ such that  \eqref{I_preserve} holds. Define the process $\{{(W^0)}_n, n \in \mathbb{N}\}$ as the channel process that 
starts with $W^0$ and evolves as in \eqref{W_n}.  We define $\{{(W^1)}_n, n \in \mathbb{N}\}$ similarly.
Furthermore,  define the two processes $H_n^0=H({(W^0)}_n)$ and $H_n^1=H({(W^1)}_n)$. We have,
\begin{align*}
& \mathbb{E}[H_{m+k+1}(1-H_{m+k+1})] \\
&\stackrel{(a)}{=} \frac{\mathbb{E}[H_{m+k}^0(1-H_{m+k}^0)] +\mathbb{E}[H_{m+k}^1 (1-H_{m+k}^1)] }{2} \\
&\stackrel{(b)}{\geq} (a_m)^k \frac{ f_m(H(W^0)) + f_m(H(W^1))}{2} \\
& \stackrel{(c)}{\geq} (a_m)^k \frac{f_m( 1-(1-H(W))^2) +f_m( H(W)^2)}{2} \\
&\stackrel{(d)}{=}  (a_m)^k f_{m+1}( H(W)) \\
& =  (a_m)^k \frac{f_{m+1}( H(W))}{f_m(H(W))} f_m(H(W)) \\
& \geq   (a_m)^k \biggl [ \inf_{h\in(0,1)}  \frac{f_{m+1}(h)}{f_m(h)} \biggr] f_m(H(W)) \\
& \stackrel{(e)}{=} (a_m)^{m+1}    f_m(H(W)).
\end{align*}  
In the above chain of inequalities, relation (a) follows from the fact that $W_m$ has $2^m$ possible outputs among which half of 
them are branched out from $W^0$ and the other half are
 branched out from $W^1$. 
 Relation (b) follows from the induction hypothesis 
given in \eqref{BMS_m}. Relation (c) follows from \eqref{ex-}, \eqref{ex+} and 
the fact that the function $f_m$ is concave. More precisely, because $f_m$ is
 concave on $[0,1]$, we have the following inequality for any sequence of 
numbers $0 \leq x' \leq x \leq y \leq y' \leq 1$ that satisfy 
 $\frac{x+y}{2} = \frac{x'+y'}{2}$:  
\begin{equation} \label{conc}
 \frac{f_m(x') + f_m(y')}{2} \leq \frac{f_m(x) + f_m(y)}{2}.
\end{equation}  
In particular, we set $x'=H(W)^2$, $x=H(W^1)$, $y=H(W^0)$, $y'=1-(1-H(W))^2$ and we know from \eqref{ex-} and \eqref{ex+} that
 $0 \leq x' \leq x \leq y \leq y' \leq 1$. Hence, by \eqref{conc} we obtain (c). Relation (d) follows 
from the recursive definition of $f_m$ given in \eqref{f_n}. Finally, 
 relation (e) follows from the definition of $a_m$ given in \eqref{a_m}.
\end{proof}

Up to now, we have provided bounds on the speed of polarization for the BEC as well as general BMS channels.   
In the final part of this section, we rigorously relate the results  obtained in previous parts to 
finite-length performance of polar codes. In other words, answering Question~\ref{Q4} stated in Section~\ref{prob_form} is the main focus for the remaining part of this section.

\subsection{Universal Bounds on the Scaling Behavior of Polar Codes}
\subsubsection{Universal Lower Bounds}
Consider a BMS channel $W$ and let us assume that a  
polar code is required with block-error probability  at most a given
 value $P_{\rm e} >0$.  One way to accomplish this is to ensure 
that the right side of \eqref{P_e}  is less than $P_{\rm e}$. 
However, this is only a sufficient condition that 
might not be necessary.   Hence, we call the right side of \eqref{P_e} 
\emph{the strong reliability condition}.  Numerical and analytical investigations (see \cite{KMTU10} and \cite{mani}) 
suggest that once the sum of individual errors in the right side of \eqref{P_e} is less than $1$, then it provides a fairly good estimate of $P_{\rm e}$. In fact, the smaller the sum is the closer it is to $P_{\rm e}$. Hence, the sum of individual errors can be considered as a fairly accurate proxy for $P_{\rm e}$.  
 Based on 
this measure of the block-error probability, we provide bounds on how the rate $R$ scales in terms of the block-length $N$. 
\begin{theorem} \label{main_univ}
For any BMS channel $W$ with capacity $I(W) \in (0,1)$, 
there exist constants $P_{\rm e},\alpha>0$, that depend only on $I(W)$, 
such that 
\begin{align} \label{sum-eps}
 \sum_{i \in \mathcal{I}_{N,R}} E(W_{N}^{(i)}) \leq P_{\rm e},
\end{align}
implies 
\begin{equation} \label{Rb}
R <I(W)- \frac{\alpha}{N^\frac{1}{\mu}}.
\end{equation}
Here, $\mu$ is a universal parameter equal to $\mu = 3.579$.
\end{theorem}
A few comments are in order: 

(i) The value of $\mu$ stated in Theorem~\ref{main_univ} (i.e. $\mu = 3.579$) can be slightly improved by the following procedure. 
As we will see shortly, we can obtain an increasing sequence
of candidates, call this sequence $\{\mu_m\}_{m\in \mathbb{N}}$,
for the universal parameter $\mu$ in \eqref{Rb}. For each $m$, in order to show
the validity of $\mu_m$, we need to verify the concavity of
a certain polynomial on $[0, 1]$ (the polynomial is defined in \eqref{f_n} with $f(z) = z(1-z)$). 
We explained in  Remark~\ref{rem_b} how we can accomplish this using the Sturm chain method.
The value of $\mu$
stated in Theorem~\ref{main_univ} is the one corresponding to $m=10$, an
arbitrary choice.  If we increase $m$, we get a new candidate for $\mu$ to plug into \eqref{Rb}, i.e., $\mu_{16}=3.614$.
We conjecture that the sequence $\mu_m$ converges to $\mu_\infty=3.627$,
the parameter for the BEC. If such  a conjecture holds, then the channel BEC polarizes the fastest among the BMS channels (see Question~\ref{Q2}). 

(ii) Let $P_{\rm e},\alpha,\mu$ be as in Theorem~\ref{main_univ}.  If we
require the block-error probability to be less than $P_{\rm e}$ (in
the sense that the condition  \eqref{sum-eps} is fulfilled), then
the block-length $N$ should be at least \begin{equation} N >
(\frac{\alpha}{I(W)-R})^\mu.  \end{equation}

(iii) From \eqref{min_block} we know that the value of $\mu$ for the random linear ensemble is $\mu=2$, which is the optimal 
value since  the
variations of the channel itself require $\mu \geq 2$. Thus, given a rate $R$, 
reliable transmission by polar codes requires a larger block-length   than the optimal 
value.

{\em Proof of Theorem~\ref{main_univ}:} To fit the bounds of Section~\ref{second_approach} into the framework of Theorem~\ref{main_univ}, let us first introduce the sequence
 $\{\mu_m\}_{m \in \mathbb{N}}$ as
\begin{equation}
 \mu_m= -\frac{1}{\log a_m},
\end{equation}
where $a_m$ is defined in \eqref{a_m} with starting function $f(z)=z(1-z)$.
From Lemma~\ref{BMS} we know that for a suitable $m$, the speed with which the quantity 
$\mathbb{E}[H_n(1-H_n)]$ decays is lower bounded by $a_m=2^{-\frac{1}{\mu_m}}$. More precisely, for 
$n\geq m$ we have $\mathbb{E}[H_n(1-H_n)] \geq 2^{-\frac{(n-m)}{\mu_m}} f_m(H(W))$.
To relate the strong reliability condition in \eqref{sum-eps} to the rate bound in \eqref{Rb}, we need the following lemma. 
\begin{lemma} \label{gamma}
Consider a BMS channel $W$ and assume that there exist positive
 real numbers $\gamma, \theta$, and $m\in \mathbb{N}$ such that $\mathbb{E}[H_n(1-H_n)] \geq \gamma 2^{-n \theta}$ for $n \geq m$. 
Let $\alpha, \beta \geq 0$ be such that $2\alpha+ \beta=\gamma$, we have for $n \geq m$
\begin{equation}
{\rm{Pr}}(H_n \leq \alpha 2^{-n \theta}) \leq I(W)-\beta2^{-n\theta}.
\end{equation}
\end{lemma}
The proof of this lemma is provided in the appendices. 
Let us now use the result of Lemma~\ref{gamma} to conclude the proof of Theorem~\ref{main_univ}.
By Lemma~\ref{BMS}, we have for  $n \geq m$
\begin{align*}
\mathbb{E}[H_n(1-H_n)] & \geq 2^{-\frac{(n-m)}{\mu_m}} f_m(H(W)). 
\end{align*}   
Thus, if we now let 
$\gamma=  2^{\frac{m}{\mu_m}} f_m (H(W))$, $\theta=\frac{1}{\mu_m}$, and $2\alpha=\beta=\frac{\gamma}{2}$,
then by using Lemma~\ref{gamma} we obtain
\begin{equation} \label{ineqtm}
\text{Pr}(H_n \leq \frac{\gamma}{4} 2^{-\frac{n }{\mu_m}}) \leq I(W)-\frac{\gamma}{2}2^{-\frac{n }{\mu_m}}.
\end{equation}
Assume that we desire to achieve a rate $R$ equal to 
\begin{equation} \label{Rm}
R= I(W)- \frac{\gamma}{4}2^{-\frac{n }{\mu_m}}. 
\end{equation}
Let $\mathcal{I}_{N,R}$ be the set of indices chosen for such a rate $R$, i.e., $\mathcal{I}_{N,R}$ includes the $2^nR$ 
indices of the sub-channels with the least value of error probability.  Define the set $A$ as
\begin{align} \label{A_def}
A= \{ i \in \mathcal{I}_{N,R}: H(W_N^{(i)}) \geq \frac{\gamma}{4} 2^{-\frac{n }{\mu_m}}\}.
\end{align}
 In this regard, note that \eqref{ineqtm} and \eqref{Rm} imply that 
\begin{equation} \label{|A|_bound}
\mid A \mid \geq \frac{\gamma}{4} 2^{n(1-\frac{1}{\mu_m})}.
\end{equation}
 As a result, by using \eqref{bounds1} and \eqref{bounds2}  we obtain for $n\geq m$
\begin{align} 
\nonumber \sum_{i \in \mathcal{I}_{N,R}} E(W_N^{(i)}) 
& \geq \sum_{i \in A} E(W_N^{(i)}) \\
& \stackrel{\eqref{bounds2}}{\geq} \sum_{i \in A} h_2^{-1}(H(W_N^{(i)}))\\
& \stackrel{\eqref{A_def}}{\geq} \sum_{i \in A} h_2^{-1}(\frac{\gamma}{4} 2^{-\frac{n }{\mu_m}})\\ 
& \geq \mid A \mid (\frac{\gamma}{4} 2^{-\frac{n }{\mu_m}})\\ 
&\stackrel{\eqref{|A|_bound}}{\geq} \frac{\gamma}{4} 2^{n(1-\frac{1}{\mu_m})} h_2^{-1}(\frac{\gamma}{4}2^{-\frac{n }{\mu_m}})\\
& \label{sum111} \geq  \frac{\gamma^2}{16} \frac{2^{n(1-2\ \frac{1}{\mu_m} )}}{8(  \frac{n}{\mu_m} + \log \frac{4}{\gamma})},
\end{align}
 where the last step follows from the fact that for $x \in [0,\frac{1}{\sqrt{2}}]$, we have $h_2^{-1}(x) \geq \frac{x}{8 \log(\frac{1}{x})}$. Thus, having 
a block-length $N=2^n$, in order to have error probability (measured by \eqref{P_e})
  less than   $ \frac{\gamma^2}{16} \frac{2^{n(1-2\ \frac{1}{\mu_m} )}}{8 (\ \frac{n}{\mu_m} + \log \frac{4}{\gamma})}$, the rate can be at
 most $I(W)-\frac{\gamma}{4}2^{-\frac{n }{\mu_m}}$. 

Finally, if we let $m=10$ (by the discussion in Remark~\ref{rem_b}, we know that $m=10$ is suitable), 
then $\mu_{10}=\frac{1}{-\log(a_{10})}=3.579$ and choosing
\begin{equation} \label{P_e_def}
P_{\rm e}= \inf_{n \in \mathbb{N}} \bigl [ \sum_{i \in \mathcal{I}_{N,R}} E(W_N^{(i)}) \bigr ], 
\end{equation}
where $R$ is given in \eqref{Rm}, then it is easy to see from \eqref{sum111} that $P_{\rm e}>0$ (since $\frac{1}{\mu_{10}}<\frac 12$). In other words, from the definition of 
$P_{\rm e}$ in \eqref{P_e_def}, we see that $P_{\rm e}$  is the infimum of a sequence of numbers. Each member of this sequence is lower bounded in \eqref{sum111}. However, it is easy to that this lower bound (and hence the sequence) diverges in $n$ (note that $\frac{1}{\mu_{10}}<\frac 12$). As a result, the value of $P_{\rm e}$, which is defined as the infimum of this sequence, is strictly positive, i.e., $P_{\rm e} > 0$. 
Furthermore, from \eqref{P_e_def}, it is easy to see that
to have the value of the sum $\sum_{i \in \mathcal{I}_{N,R}} E(W_N^{(i)})$ to be less than $P_{\rm e}$, the rate should be less than $R$ given in \eqref{Rm}.

\subsubsection{Universal Upper Bounds}
In this part, we provide upper bounds on the block-length $N$ for polar codes, in terms of the rate $R$, that is required to obtain an error probability less than a given value $P_{\rm e}$ (see Question~\ref{Q4} in Section~\ref{prob_form}). Again, the key component here is the upper-bounds on the speed of polarization, e.g. the bounds derived in Table~\ref{nb_n} for the BEC and the universal bound \eqref{univ-up}. 
\begin{theorem}\label{thm_scaling_up}
Let $Z_n=Z(W_n)$ be the Bhattacharyya process associated to a BMS channel $W$. Assume that for
$n \in \mathbb{N}$ we have
\begin{equation} \label{speed_up}
\mathbb{E}[(Z_n (1-Z_n) )^\alpha] \leq \beta 2^{-\rho n},
\end{equation}
where $\alpha, \beta, \rho$ are positive constants and $\alpha<1$.  Then, the block-length $N$ required   to achieve an error probability $P_{\rm e}>0$ at a given rate $R<I(W)$ is bounded from above by
\begin{equation}\label{1logNbound}
\log N \leq  (1+ \frac{1}{\rho}) \log \frac{1}{d} + c(\log(\log \frac {4}{d}))^2 ,
\end{equation}  
where $d=I(W)-R$ and $c$ is a universal positive constant that depends on $\alpha, \beta, \rho, P_{\rm e}$.
\end{theorem}
Before proceeding with the proof of Theorem~\ref{thm_scaling_up}, let us note a few comments:  

(i) In
the previous sections we have computed several candidates for the value $\rho$ required in Theorem~\ref{thm_scaling_up}.  
As an example,  using the universal candidate for $\rho$ given  in \eqref{univ-up} (i.e., $\rho =0.202$), we obtain the following corollary.
\begin{corollary}
For any BMS channel $W$,  the block-length $N$ required to achieve a rate $R < I(W)$ scales at most as
\begin{equation}
N \leq \Theta \bigl ( \frac{1}{(I(W) - R)^{6}}\bigr).
\end{equation} 
\end{corollary}
One important consequence of this corollary is that polar codes require a block-length that 
scales polynomially in terms of the reciprocal of gap to capacity.\footnote{The fact that polar codes need 
a polynomial block-length in terms of the reciprocal of the gap to capacity is also proven in the 
recent independently-derived result of \cite{patrick-venkat}.}

(ii) As we will see in the proof of Theorem~\ref{thm_scaling_up}, the result of this theorem is also valid if we replace $P_{\text{e}}$ 
with the sum of Bhattacharyya values of
the channels that correspond to the good indices (this sum is indeed an upper bound for $P_{\rm{e}}$).

{\em{Proof of Theorem~\ref{thm_scaling_up}:}}    
Throughout the proof we will be using two key lemmas (Lemma~\ref{1aux1} and Lemma~\ref{1aux2}) that are stated in the appendices.  Let 
\begin{equation}
d= I(W)-R.
\end{equation}
We define $n_0 \in \mathbb{N}$ to be
\begin{equation}\label{1n_0}
n_0=\biggl \lceil \frac{1}{\rho} \log \frac{3 (1+c_1)(1+2c_2 c_3 \beta) }{d} \biggr \rceil,
\end{equation}  
where $\beta$ is given in \eqref{speed_up} and the constants $c_1$, $c_2$ and $c_3$ are given in Lemmas \ref{1aux1}, \ref{1aux2} and \ref{1aux3}, respectively. 
As a result of Lemma~\ref{1aux1} and \eqref{1n_0}, we have for $n \geq n_0$ 
\begin{align} \nonumber
{\rm{Pr}} ( Z_n \leq \frac 12) & \geq I(W) - c_1 2^{-n \rho} \\ \nonumber
& \stackrel{(a)}{\geq} I(W) - \frac{d}{3} \\  \label{pZ<nim}
& = R + \frac{2}{3} d,
\end{align} 
where step (a) is a consequence of \eqref{1n_0} that for $n \geq n_0$ we have $c_1 2^{-n \rho} \leq \frac{d}{3}$.
We now define the set $\mathcal{A}$ as follows. Let $N_0=2^{n_0}$ and
\begin{equation} \label{1Bdef}
\mathcal{A} =\bigl \{i \in \{0, \cdots, N_0-1\}: Z(W_{N_0}^{(i)}) \leq \frac 12 \bigr\}.
\end{equation}
In other words $\mathcal{A}$ is the set of indices at level $n_0$ of the corresponding infinite binary tree of $W$ (see Section~\ref{chanpol}) whose
Bhattacharyya parameter is not so large.  Also, from \eqref{pZ<nim} the set $\mathcal{A}$ contains more than a fraction $R$ of all the sub-channels at level $n_0$. The idea is then to go further down through the infinite binary  tree at a level $n_0+n_1$ (the value of $n_1$ will be specified shortly).  We  then observe that the sub-channels at level $n_0+n_1$ 
that are branched out from the set $\mathcal{A}$ are polarized to a great extent in the sense that sum of their Bhattacharyya parameters is below $P_{\rm e}$ (see Figure~\ref{fig:n0n1} for a schematic illustration of the idea).  
 \begin{figure}[ht!]
\begin{center} 
 \includegraphics[width=8cm]{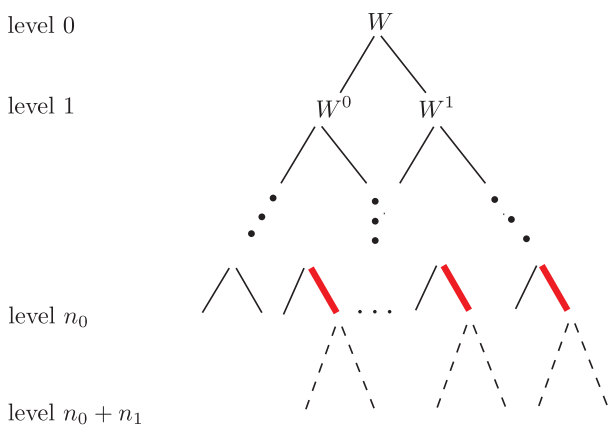}
\end{center}
\caption{The infinite binary tree of channel $W$. The red (also bold) edges at level $n_0$ of this tree correspond to the sub-channels at level $n_0$ whose Bhattacharyya
parameter is less that $\frac 12$ (i.e., the set $\mathcal{A}$). The idea is then to focus on these red (bold) indices.  We consider  the sub-channels that are branched out from these red indices at a level $n_0+n_1$ (as shown in the figure). By a careful choice of $n_1$, we observe that many of these specific sub-channels at level $n_0+n_1$ are greatly polarized in the sense that sum of their Bhattacharyya parameters is less than $P_{\rm e}$.  We also show that the fraction of these sub-channels is larger than $R$.  }\label{fig:n0n1}
\end{figure}        

We proceed by finding a suitable candidate for $n_1$.  Our objective is to choose $n_1$ large enough s.t. there is a set of indices at level $n_0+n_1$ with the following properties: (i) sum of the Bhattacharyya parameters of the sub-channels in this set  is less than $P_{\rm e}$ and (ii) the cardinality of this set is at least $R 2^{n_0+n_1}$.   In what follows, we will first use the hypothesis of Lemma~\ref{1aux2} to give a candidate for $n_1$ and then we make it clear that such a candidate is suitable for our needs.  
Let $\{B_m\}_{m \in \mathbb{N}}$ be a sequence 
of iid Bernoulli($\frac12$) random variables. We let $n_1$ be the smallest positive integer such that the following holds
\begin{equation}\label{sumpowBeps}
{\rm Pr}(2^{-2^{\sum_{i=1}^{n_1} B_i}} \leq \frac{P_{\rm e}}{2^{n_0+n_1}}) \geq 1- \frac{d}{6}.
\end{equation}  
It is easy to see that \eqref{sumpowBeps} is equivalent to 
\begin{equation}\label{sumBeps0}
{\rm Pr}\bigl(\sum_{i=1}^{n_1} B_i \geq \log ( \log \frac{1}{P_{\rm e}} + n_0+n_1 ) \bigr) \geq 1- \frac{d}{6}.
\end{equation}  
Now, note that we can write
\begin{align} \nonumber
& \log ( \log \frac{1}{P_{\rm e}} + n_0+n_1 ) \\ \nonumber
& = \log ( 1 + \log \frac{1}{P_{\rm e}} + 1 + n_0 + n_1 -2) \\ \nonumber
& \stackrel{(a)}{\leq} \log(1+ \log \frac{1}{P_{\rm e}}) + \log (1 + n_0 + n_1 -2) \\ \label{concav_log}
& \leq \log(\log \frac{2}{P_{\rm e}}) + \log (n_0 + n_1), 
\end{align}
where (a) follows from the fact that the function $f(x) = \log(1+x)$ is a concave function with $f(0) = 0$, and for any such function the following is true:
$f(x+ y) \leq f(x) + f(y), \forall x,y \geq 0$. As a result of \eqref{sumBeps0} and \eqref{concav_log}, in order for \eqref{sumpowBeps} to hold the following 
is sufficient:
\begin{equation}\label{sumBeps}
{\rm Pr}\bigl(\sum_{i=1}^{n_1} B_i \geq \log ( \log \frac{2}{P_{\rm e}}) + \log( n_0+n_1 ) \bigr) \geq 1- \frac{d}{6}.
\end{equation}  
Also, as the random variables $B_i$ are  Bernoulli($\frac 12$) and iid, the relation \eqref{sumBeps} is equivalent to
\begin{equation}\label{suffsumB}
\frac{\displaystyle \sum_{j=0}^{\log(\log \frac{2}{P_{\rm e}}) + \log(n_0+n_1)}  \binom{n_1}{j} }{2^{n_1}} <
  \frac{d}{6}.
\end{equation} 
A sufficient condition for \eqref{suffsumB} to hold is as follows:
\begin{equation*}
\frac{ n_1 ^{1+\log(\log \frac{2}{P_{\rm e}}) + \log(n_0+n_1)}}{2^{n_1}}
\leq   \frac{d}{6},
\end{equation*}
 and after applying the function $\log(\cdot)$ to both sides and some further simplifications we reach to
 \begin{equation} \label{n1log}
n_1- \bigl({1+\log(\log \frac{2}{P_{\rm e}}) + \log(n_0+n_1)} \bigr) \log n_1 \geq \log \frac{6}{d}.
 \end{equation}  
It can be shown through some simple steps that there is a constant $ c_6>0$ (that also depends on $P_{\rm e}$) s.t. if we choose 
 \begin{equation}
 n_1= \biggl \lceil \log \frac{6}{d} + c_6 (\log(\log \frac {6}{d}))^2 \biggr \rceil,
 \end{equation}
 then the inequality \eqref{n1log} holds. Now, let $\tilde{N}=2^{n_0+n_1}$ and consider the set $\mathcal{A}_1$ defined as
 \begin{equation}\label{1B1def}
\mathcal{A}_1 =\bigl \{i \in \{0, \cdots, \tilde{N}-1\}: Z(W_{\tilde{N}}^{(i)}) \leq \frac{P_{\rm e}}{\tilde{N}}  \bigr\}.
\end{equation}
We now show that 
\begin{equation}  \label{1B1card}
\frac{ | \mathcal{A}_1 |}{\tilde{N}} \geq R.
\end{equation}  
This relation together with \eqref{1B1def} shows that block error probability of the  polar code of block-length $\tilde{N}$ and rate $R$ is at most $P_{\rm e}$. 

 In order to show \eqref{1B1card}, we consider the sub-channels  in $\mathcal{A}_1$ that  
are branched out from the ones in the set $\mathcal{A}$ (defined in \eqref{1Bdef}). 
Let $i\in \mathcal{A}$ and consider the sub-channel $W_{N_0}^{(i)}$. 
At level $n_0 + n_1$ there are in total $2^{n_1}$ sub-channels that branch out from the sub-channel $W_{N_0}^{(i)}$ (which is itself at level $n_0$). By using \eqref{Z_nextreme} it is easy to see that the process $Z_n$ fulfills the condition \eqref{1genericX_n} of Lemma~\ref{1aux2}.  From Lemma~\ref{1aux2}, relation \eqref{sumpowBeps}, 
and the fact that for any two events $A$ and $B$ we have $\text{Pr}(A \cap B) \geq  \text{Pr}(A) + \text{Pr}(B) -1$, we obtain the following: 
At level $n_0+n_1$, there are in total $2^{n_1}$ 
sub-channels that are branched out from $W_{N_0}^{(i)}$, and among these sub-channels, a fraction at least
\begin{equation*}
1- \frac{d}{6} - c_2Z(W_{N_0}^{(i)})(1+ \log \frac{1}{Z(W_{N_0}^{(i)})}),
\end{equation*}  
have Bhattacharyya value less than $\frac{P_{\rm e}}{\tilde{N}}$. Therefore, the number of channels at  
level $n_0+n_1$ that are branched out from $W_{N_0}^{(i)}$ and have Bhattacharyya value less than $\frac{P_{\rm e}}{\tilde{N}}$  
is at least 
\begin{equation*}
2^{n_1} \bigl (1- \frac{d}{6} - c_2Z(W_{N_0}^{(i)})(1+ \log \frac{1}{Z(W_{N_0}^{(i)})})  \bigr ).
\end{equation*}  
Hence,  the total number of sub-channels at level $n_0+n_1$ that are branched out from a sub-channel in $\mathcal{A}$ and have Bhattacharyya value less that $\frac{P_{\rm e}}{\tilde{N}}$ is at least
\begin{equation} \label{lowerB1} 
2^{n_1}\sum_{i \in \mathcal{A}} \bigl (1-\frac{d}{6} - c_2Z(W_{N_0}^{(i)})(1+ \log \frac{1}{Z(W_{N_0}^{(i)})})  \bigr ).
\end{equation}
We can further write
\begin{align*}
\eqref{lowerB1} = 
2^{n_1} (1-\frac{d}{6}) | \mathcal{A} |  - c_2 2^{n_1}\sum_{i \in \mathcal{A}}  Z(W_{N_0}^{(i)})(1+ \log \frac{1}{Z(W_{N_0}^{(i)})}) .
\end{align*}
Now, by using \eqref{pZ<nim} and \eqref{1Bdef} we have $|\mathcal{A}| \geq 2^{n_0} (R+\frac 23 d)$, and hence \eqref{lowerB1} can be lower bounded by
 \begin{equation}
 2^{n_0+n_1}   \bigl( \! (R + \frac 23 d)(1 - \frac{d}{6} ) - c_2 2^{-n_0}\!\! \sum_{i \in \mathcal{A}} \!  Z(W_{N_0}^{(i)}) (1+ \log \frac{1}{Z(W_{N_0}^{(i)})}) \! \bigr). \label{lowerB11}
\end{equation}
 We further have
\begin{align*}
&c_2 2^{-n_0} \sum_{i \in \mathcal{A}}  Z(W_{N_0}^{(i)})(1+ \log \frac{1}{Z(W_{N_0}^{(i)})}) \\
& \stackrel{(a)}{\leq}  2c_2 2^{-n_0} \sum_{i \in \mathcal{A}}  Z(W_{N_0}^{(i)})\log \frac{1}{Z(W_{N_0}^{(i)})} \\
& \stackrel{\text{Lemma~\ref{1aux3}}}{\leq}  2c_2 c_3 2^{-n_0} \sum_{i \in \mathcal{A}} (Z(W_{N_0}^{(i)}) (1-Z(W_{N_0}^{(i)})))^\alpha \\
& \leq 2c_2 c_3 \mathbb{E}[(Z_{n_0} (1-Z_{n_0}))^\alpha] \\
&\stackrel{\eqref{speed_up}}{\leq} 2c_2 c_3 \beta 2^{-n_0 \rho}\\
& \stackrel{\eqref{1n_0}}{\leq} \frac{d}{3},
\end{align*}    
where (a) follows from the fact that for $x \leq \frac 12$ we have $1 + \log \frac{1}{x} \leq 2 \log \frac{1}{x}$.
Therefore, the expression \eqref{lowerB11} (and hence \eqref{lowerB1}) is lower-bounded by
\begin{equation*}
2^{n_0+n_1} \bigl (  (R + \frac {2}{3} d) (1-\frac{d}{6}) - \frac{d}{3} \bigr ) \geq 2^{n_0+n_1} R = \tilde{N} R.
\end{equation*} 
Hence, the relation \eqref{1B1card} is proved and 
a block-length of size $\tilde{N}$ is sufficient to achieve a rate $R$ and error at most $P_{\rm e}$. It is now easy to see that $\log \tilde{N} = n_0 + n_1$ has the form of \eqref{1logNbound}.

\section{Conclusion} \label{conclusion}
Let us briefly summarize our main results and discuss some interesting avenues for
future research. 

We have considered the tradeoff between the rate and the block-length
for a fixed error probability when we use polar codes and the
successive cancellation (SC) decoder.  For a BMS channel $W$,
consider the setting where we require the error probability (measured
by the sum of the Bhattacharyya parameters) to be a fixed value
$P_{\rm e} > 0$. We have shown that in this setting the block-length $N$
scales in terms of the rate $R < I(W)$ as $N \geq
\frac{\alpha}{(I(W)-R)^{\underline{\mu}}}$, where $\alpha$ is a
positive constant that depends on $P_{\rm e}$ and $I(W)$, and
$\underline{\mu}  = 3.579$.  In other words, the required block-length
$N$ is at least $\Theta \bigl (\frac{1}{(I(W)-R)^{\underline{\mu}}}
\bigr)$.  A comparison with \eqref{min_block} indicates that polar
codes require a larger block-length compared to the best possible
codes (for which $\mu=2$). This provides an analytical explanation
for the rather long blocklenghts which are required in numerical
experiments involving polar codes.

In the same setting, we have also derived an upper bound on the
required blocklenght by showing that $N \leq
\frac{\beta}{(I(W)-R)^{\overline{\mu}}}$, where $\beta$ is a constant
that depends on $P_{\rm e}$ and $I(W)$, and $\overline{\mu}=6$.  In
other words, the required block-length is at most   $\Theta \bigl
(\frac{1}{(I(W)-R)^{\overline{\mu}}} \bigr)$.

We conjecture that the value of $\underline{\mu}$ can be increased
up to $\underline{\mu} = 3.627$ (the corresponding parameter for
the BEC).  In the same vain, the value of $\overline{\mu}$ can be
decreased below $\overline{\mu}  = 6$ by searching for better
candidates for the function $g(\cdot)$ with a smaller $L_g$ (see
\eqref{L_g}). Indeed, in a follow up work \cite{GB13}, such functions
are constructed by carefully evolving a suitable sequence of
candidates $g_m(\cdot)$ through the various polarization levels
$m$.  In this way, a new scaling bound with $\overline{\mu} = 5.77$
is obtained.   

In view of our results, perhaps the most important open question,
both from the theoretical as well as the practical side, is to
improve the finite-length performance of these codes.  We can
approach this problem from two perspectives: (i) by devising better
decoding algorithms and (ii) by changing the construction of polar
codes (e.g., by concatenating them with other codes, use other
polarizing kernels, etc).  In any attempt to improve the finite-length
performance, one main objective should be to improve the scaling
exponent (or  the speed of polarization).

In \cite{TV11}, the authors  combine both of these perspectives and
provide experimental evidence that the short-length performance of
polar codes can be improved considerably. More precisely, a
successive-cancellation \textit{list} decoder (SCL) is proposed in
\cite{TV11} to boost the performance of the SC decoder to that of
the MAP decoder. However, even under MAP decoding the performance of polar codes is
still not competitive.  Hence, by a simple concatenation  with
a very high-rate code, the MAP performance is improved to a
great extent.  The main issue of the successive cancellation list
decoder is its memory consumption which scales linearly with the
list-size. There are by now various other techniques to improve the
finite-length performance of polar codes. For a partial list see
\cite{SH}-\cite{MHU14}. It is also an interesting open question to
find out how the scaling exponent of the coding method of \cite{TV11}
changes with the list-size parameter. For a fixed finite list-size,
it is proven in \cite{MHU13} that the scaling exponent does not
change compared to original polar codes when we use the MAP decoder.  We believe that the methods
developed in this paper can be useful in this regard.

Another approach is to consider polar codes with general $\ell \times \ell$ kernels with the hope that polar codes with larger kernels might  have a better finite-length behavior. The related discussions in \cite[Chapter 1]{hamedthesis} support the fact that when $\ell$ grows large, for almost any  kernel, the scaling exponent ($\mu$) of the associated polar code tends to $\frac 12$.
Recall from  \eqref{min_block} that the optimal value of $\mu$ over all the codes is $\frac 12$, and for  polar codes (with $\ell=2$)  the scaling exponent is at most $\mu=\frac{1}{3.6} \approx 0.27$.  
We keep in mind that, in general, the decoding complexity of (extended) 
polar codes is $O(2^\ell N \log N)$, where $N$ is the block-length.
An interesting question here is to find suitable $\ell \times \ell$ kernels with a better scaling exponent 
than the $\ell =2$ case, as well as a reasonable complexity.

Finally, let us note that all these scaling results are in principle extendable to further applications of 
polarization theory and polar codes in various other scenarios (see e.g. \cite{GB13}).

\section*{Acknowledgment}
The authors wish to thank Erdal Ar\i kan, Alexander Barg, Mani Bastani-Parizi, Ali Goli, Marco Mondelli, and Emre Telatar for their
valuable comments on this topic.


\begin{appendices}

\section{Proofs} \label{proofs}

\subsubsection{Proof of Lemma~\ref{lim:equal}}
The proof of \eqref{33} is an easy application of the Markov inequality:
We have
\begin{equation} \label{f_log_g}
 { \rm{Pr}}(Z_n \in [a,b]) \leq { \rm{Pr}}(f(Z_n) \geq \min_{z \in [a,b]} f(z)) \leq  \frac{\mathbb{E}[f(Z_n)]}{\displaystyle \min_{z \in [a,b]} f(z)},
\end{equation}
and \eqref{33} follows by applying $\frac{1}{n} \log(\cdot)$ to \eqref{f_log_g}.

To prove \eqref{32}, we define 
 sequences $\{x_n\}_{n \geq 1}$ and $\{y_n\}_{n \geq 1}$ as 
 \begin{align}
  &  \label{a_n} x_n=2^{-n},\\
  &  \label{b_n} y_n=1-2^{-n}.
 \end{align}
We start by noting that
\begin{align*}
 \mathbb{E} (Z_n(1-Z_n)) & \leq   
\sum_{i=1}^n 2^{-i} \text{Pr} (Z_n \in [x_{i+1},x_i ]) \\
& \;\;\; +  \sum_{i=1}^n 2^{-i} \text{Pr} (Z_n \in [y_{i},y_{i+1}])\\
& \;\;\;+ 2^{-n}.
\end{align*}
As a result, there exists an index $j \in \{1, \cdots,n\}$  such that at least one of the following cases 
occurs:
\begin{align}
& \label{casex}\mathbb{E} [Z_n(1-Z_n)] \leq 2n \bigl [2^{-j} \text{Pr} (Z_n \in [x_{j+1},x_{j}])  \bigr] +2^{-n}, \\
& \text{or} \nonumber \\
& \label{casey}\mathbb{E} [Z_n(1-Z_n)] \leq  2n \bigl[ 2^{-j}\text{Pr} (Z_n \in [y_{j},y_{j+1}])  \bigr]+2^{-n}.
\end{align}
We show that in each of these cases the statement of the lemma holds. 
Note further that because of the symmetry of $Z_n$ we can write 
\begin{align*}
 \text{Pr} (Z_n \in [y_{j+1},y_{j}] \mid Z_0 =z)  = \text{Pr} (Z_n \in [x_{j+1},x_{j}] \mid Z_0=1-z), 
\end{align*}
Hence, without loss of generality we can assume that \eqref{casex} holds.
The proof consists of two parts:

We first assume that $a=1-b=\frac 14$ and prove \eqref{32} for this choice of $a,b$. 
\begin{lemma} \label{lem_clain_i}
For any $ j \in \{1,\cdots, n\}$ we have
\begin{equation} \label{claim_i}
2^{-j} \text{Pr}(Z_n \in [x_{j+1},x_j]) 
\leq (n+1)\text{Pr}(Z_{n} \in [\frac 14, \frac 34]) + 2^{-n}.
\end{equation}  
\end{lemma}
The proof of this lemma will appear shortly. But before that, 
we note that by using the result of this lemma and \eqref{casex} we obtain
\begin{align*} 
\mathbb{E} (Z_n(1-Z_n)) \leq 2n(n+1) \bigl[\text{Pr}(Z_n \in [\frac 14, \frac 34]) \bigr] + (2n+1)2^{-n} ,
\end{align*}
and as a result, by taking $\frac{1}{n} \log(\cdot)$ from both sides, \eqref{32} is 
proved for $a=1-b=\frac 14$.  

Now, for other choices of $a,b \in (0,1)$ s.t. $\sqrt{a} \leq1-\sqrt{1- b}$ we can proceed as follows.
Let us first recall the definition of the maps $t_0,t_1$ from \eqref{T} as well as the maps $\phi_{\omega_n}$ from Definition~\ref{phi}.
Also, let $p_n(z,a,b)$ be defined as in \eqref{p_n}. 
 We have
\begin{align*} 
& p_{n+1}(z,a,b) =
  \sum_{\omega_{n+1} \in \Omega_{n+1}} \frac{1}{2^{n+1}} \mathbbm{1}_{\{z \in \phi_{\omega_{n+1}} ^{-1} [a,b]\}}
\\&=\sum_{\omega_n \in \Omega_n} \frac{1}{2^n} \frac{\mathbbm{1}_{\{z \in \phi_{\omega_n} ^{-1} [t_0^{-1}(a),t_0^{-1}(b)]\}}+\mathbbm{1}_{\{z \in \phi_{\omega_n} ^{-1} [t_1^{-1}(a), t_1^{-1}(b)]\}}}{2}\\
&= \frac 12 \bigl(p_n(z,t_0^{-1}(a), t_0^{-1}(b))+p_n(z,t_1^{-1}(a), t_1^{-1}(b)) \bigr).
\end{align*}
It is easy to see that if $\sqrt{a}\leq 1-\sqrt{1-b}$, then 
\begin{equation*}
[t_0^{-1}(a),t_1^{-1}(b)]  \subseteq [t_0^{-1}(a), t_0^{-1}(b)] \cup [t_1^{-1}(a), t_1^{-1}(b)],
\end{equation*}
and hence,
\begin{equation*}
2 p_{n+1} (z,a,b) \geq p_n(z,t_0^{-1}(a),t_1^{-1}(b)).   
\end{equation*}
Continuing this way, we can show that for $m \in \mathbb{N}$
\begin{multline} \label{1ineq_m}
2^{m} p_{n+m} (z,a,b) \\
\geq p_n(z,\overbrace{t_0^{-1} \circ \cdots \circ t_0^{-1}}^{m \text{ times}} (a),
\overbrace{t_1^{-1} \circ \cdots \circ t_1^{-1}}^{m \text{ times}} (b)).   
\end{multline}
As $m $ grows large, we have 
\begin{align*}
& \overbrace{t_0^{-1} \circ \cdots \circ t_0^{-1}}^{m \text{ times}} (a) \to 0,\\
& \overbrace{t_1^{-1} \circ \cdots \circ t_1^{-1}}^{m \text{ times}} (b) \to 1.
\end{align*}
Therefore, by \eqref{1ineq_m} there exists a positive integer $m_0 \in \naturals$ that only depends on $a,b$ and for $n \in \mathbb{N}$ and $z \in [0,1]$
\begin{equation*}
2^{m_0} p_{n+m_0}(z,a,b) \geq p_n(z, \frac 14,\frac 34 ). 
\end{equation*}
The proof of \eqref{32} now follows from this relation together with \eqref{claim_i} and the result of Lemma~\ref{a_bec}. It remains to prove Lemma~\ref{lem_clain_i}.

\emph{Proof of Lemma~\ref{lem_clain_i}:}  Consider the relation \eqref{claim_i} for $ 1 \leq j \leq n$. If $j=1$, 
then there is nothing to prove. Hence, in the following we assume that $2 \leq j \leq n$. 
We  prove that for any fixed $j$, such that  $ 2 \leq j \leq n$, the claim of \eqref{claim_i} holds true.  
So let us fix the index $j$ and prove \eqref{claim_i}  for any value of $n \in \mathbb{N}$.
The proof consists of two steps.

\emph{Step 1:}  We first show that $ \forall m \in \mathbb{N}$,
\begin{equation} \label{claim_m}
 {\rm{Pr}}(Z_m \in [x_{2j+1}, x_j])
 \leq m \text{Pr}(Z_m \in [x_j, \frac 34]) + \frac{1}{2^m}.
\end{equation}
To prove \eqref{claim_m}, fix $m \in \mathbb{N}$ and define the sets $A$ and $B$ as
\begin{align*}
 &A= \{ (b_1, \cdots,b_m)  \in \Omega_m : t_{b_m} \circ  \cdots \circ t_{b_1}(z) \in [x_{2j+1}, x_j] \}.\\
 &B= \{ (b_1, \cdots,b_m) \in \Omega_m :  t_{b_m} \circ \cdots \circ t_{b_1}(z) \in [x_j, \frac 34] \}.
\end{align*}
 In words, $A$ is the set of all the paths that start from $z=Z_0$ 
and end up in $[x_{2j+1},x_j]$ and $B$ is the set of paths that start from $z$
and end up in $[x_j, \frac 34]$. 
Consider the sets $A_{k}$, $k \in \{1,\cdots,m\}$, 
defined as
\begin{equation}
 A_{k}=\{(b_1, \cdots,b_m) \in A: b_k=1 ;    b_i=0 \,\,\,  \forall i >k \}.
\end{equation}
It is easy to see $A_k$'s are disjoint and 
\begin{equation} \label{A-partition}
 \mid A- \cup_{k} A_{k} \mid \leq 1.
\end{equation} 
Our aim is now to show that for  $k \in \{1,\cdots,m\}$, 
\begin{equation} \label{A_k<B}
\mid A_{k} \mid \leq \mid B \mid.  
\end{equation}
Before proving \eqref{A_k<B}, let us show how the relation \eqref{claim_m} follows 
from \eqref{A-partition} and \eqref{A_k<B}. We have
\begin{align*}
 {\rm{Pr}}(Z_m \in [x_{2j+1}, x_j]) 
 & = \frac{\mid A \mid}{2^m} \\
 &\stackrel{\eqref{A-partition}}{ \leq} \frac{\sum_{i=1}^m  \mid A_i \mid \, +1}{2^m}\\
 &\stackrel{\eqref{A_k<B}}{\leq} m \frac{ \mid B \mid}{2^m} + \frac{1}{2^m}\\
 &=m \text{Pr}(Z_m \in [x_j, \frac 34]) + \frac{1}{2^m}.
\end{align*}
It thus remain to prove \eqref{A_k<B} and Step 1 is over. We show that there exists a one-to-one correspondence 
between  $A_{k}$ and a subset of $B$. In other words, 
we claim that we can map each member of $A_{k}$ to a distinct member of $B$. In this way, the relation \eqref{A_k<B} is immediate. 
Consider $(b_1, \cdots, b_m) \in A_{k}$. We now construct a  
distinct member $(b'_1, \cdots, b'_m) \in B$ corresponding 
to $(b_1, \cdots, b_m)$. We first 
set $b'_i=b_i$ for $i < k$ and hence the uniqueness condition is fulfilled (i.e., the choice of $b'_i=b_i$ for $i < k$  guarantees that the mapping from $A_k$ to $B$ is an invertible mapping).
Consider the number $y$ defined as
\begin{equation} \label{x_def}
y=  \left\{
\begin{array}{lr}
z &  ; \text{if } k=1  ,\\
t_{b_{k-1}} \circ \cdots \circ t_{b_1}(z) &  ; \text{if } k>1.
\end{array} \right.
\end{equation}
Note that since $(b_1, \cdots, b_m) \in A_{k}$ we have
\begin{equation} \label{circ_k}
 t_{b_m} \circ \cdots \circ t_{b_{k}}(y) \in [x_{2j+1}, x_j].
\end{equation}
Now, note that as $(b_1, \cdots, b_m) \in A_{k}$, we have $b_k=1$ and $b_i=0$ for $i>k$. Thus, in this setting  \eqref{circ_k} becomes  
\begin{equation} \label{circ_k1}
 \overbrace{t_{0} \circ \cdots \circ t_{0}}^{m-k \text{ times}}(y^2) \in [x_{2j+1}, x_j].
\end{equation}
Hence,
\begin{align} \label{x-x_j}
x_{2j+1} \leq 1-(1-y^2)^{2^{m-k}} \leq  x_j.
\end{align}
From the left side of \eqref{x-x_j} and by using Bernoulli's inequality
\begin{align*}
1-(1-x)^{\beta} \leq \beta x, \text{ where } \beta \geq 1 \text{ and } x \in [0,1], 
\end{align*}
we obtain
\begin{equation}\label{x-down}
x_{2j+1} \leq 2^{m-k} y^2 \Rightarrow 2^{-j+ \frac{k-m-1}{2}} \leq y.
\end{equation} 
From the right side of \eqref{x-x_j}  we have
\begin{equation*}
\ln (1-x_j) \leq 2^{m-k} \ln (1-y^2),
\end{equation*}
and by using the inequality 
\begin{align*}
-x-\frac{x^2}{2} \leq \ln (1-x) \leq -x,  \text{ where } x \in (0,1),
\end{align*}
we obtain 
\begin{equation} \label{x-up}
y \leq 2^{\frac{-j}{2}+\frac{k-m+1}{2}}.
\end{equation}
Let us recall that we let $b'_i=b_i$ for $i<k$ (and this makes the mapping from $A_k$ to $B$ an invertible mapping). 
We now construct the remaining values $b'_{k}, \cdots, b'_m$ by the following algorithm:  
 Consider the number $y$ given in \eqref{x_def}. In the following,  we will also construct a sequence $y=y_{k-1}, y_k, y_{k+1}, \cdots, y_m$ such 
 that for $i \geq k$ we have $y_i= t_{b'_{i}} (y_{i-1})$.   Begin with 
the initial value $y_{k-1}=y$  and for $i\geq k$ recursively construct 
$b'_i$ and $y_i$ from $y_{i-1}$ by the following rule: if 
$t_{0}(y_{i-1}) \leq \frac 34$, then $b'_{i}=0$ and $y_{i}=t_0(y_{i-1})$, 
otherwise $b'_{i}=1$ and $y_{i}=t_1(y_{i-1})$. 
We show that by this construction the value of $y_m$ would always fall in the interval $[ x_j, \frac 34]$. 
In this regard, an important observation is that
for $i$ s.t. $k-1 \leq i \leq m$, once the value of $y_{i}$  
lies in the interval $[ x_j, \frac 34]$, then
 for all $i \leq t \leq m$ we have $y_t \in [ x_j, \frac 34]$ (this is clear from construction rule of $y_t$). 
Hence, we only need to show that by the above algorithm, the 
exists an index $i$, s.t. $k-1 \leq i \leq m$, and the 
value of $y_i$  lies inside the interval $[ x_j, \frac 34]$. On the one hand, observe that due 
to \eqref{x-up} and the fact that $j\geq 2$, we have $y\leq 2^{-\frac 12}<\frac 34$.
Thus, the value of
$y_i$ is definitely less than $\frac 34$ for $i\geq k$. If the value of $y_{k-1}$ is also greater than $x_{j}$ then we have nothing to prove. 
 Else, it might be the case that 
$ y <   x_j$. We now prove that in this case the algorithm moves in a way that the value of 
$y_m$ falls eventually in the desired region $[ x_j, \frac 34]$. To show this, a moment of thought reveals that this is equivalent to
showing that we always have 
\begin{equation} \label{show}
\overbrace{t_{0} \circ \cdots \circ t_0}^{m-k+1 \text{ times }} (y) = 1-(1-y)^{2^{m-k+1}}  \geq x_{j}.  
\end{equation}
In order to have \eqref{show} it is equivalent that
 \begin{equation*}
 2^{m-k+1} \ln(1-y) \leq \ln (1-x_j),
 \end{equation*}
and after some further simplification using the inequality $ -x-\frac{x^2}{2} \leq \ln (1-x) \leq -x$, we deduce that a sufficient condition to have \eqref{show} is
\begin{equation}
x_j \leq 2^{m-k} y \Rightarrow 2^{-j+k-m} \leq y. 
\end{equation} 
But this sufficient condition is certainly met by considering the inequality \eqref{x-down} and noting the fact that $-j+ \frac{k-m+1}{2} \geq -j+k-m$ (recall that $k \leq m$).
Hence, the claim in \eqref{A_k<B} is proved and as a result,  the claim in \eqref{claim_m} is true.

\emph{Step 2:} Firstly note that in order for $Z_n$ 
to be in the interval $[x_{j+1},x_j]$, the value of $Z_{n-j}$ should certainly
lie somewhere in the interval $[x_{2j+1},  x_j^{2^{-j}}]$.
As a result, we can write
\begin{align} \nonumber
& \text{Pr} (Z_{n} \in [x_{j+1}, x_j]) \\   \nonumber
&= \text{Pr} (Z_{n} \! \in [x_{j+1}, x_j] \mid Z_{n-j} \!  \in [x_{2j+1},x_{j}])\\  \nonumber
& \,\,\, \,\,\, \,\,\,\,\,\,\,\,\,\,\,\,  \,\,\, \,\,\, \,\,\,\,\,\,\,\,\,\,\,\,   \,\,\, \,\,\, \,\,\,\,\,\,\,\,\,\,\,\, 
\times \text{Pr}(Z_{n-j} \! \!  \in [x_{2j+1},  x_{j}]) \nonumber \\ 
&  \,\,\, \,\,\, + \text{Pr} (Z_{n} \! \in \! \! [x_{j+1}, x_j] \mid \! Z_{n-j}  \nonumber
\in \! ( x_{j},x_j^{2^{-j}} ] ) \\ 
&  \,\,\, \,\,\, \,\,\,\,\,\,\,\,\,\,\,\,  \,\,\, \,\,\, \,\,\,\,\,\,\,\,\,\,\,\,  \,\,\, \,\,\, \,\,\,\,\,\,\,\,\,\,\,\, 
\times \text{Pr}(Z_{n-j} \! \!  \in (x_{j}, x_j^{2^{-j}}]), \label{n-i1}
\end{align}
and by letting $m=n-j$ in relation \eqref{claim_m}, we can easily obtain
\begin{equation} \label{n-i2}
 \text{Pr}(Z_{n-j} \in [x_{2j+1},x_{j}]) 
\leq n \text{Pr}(Z_{n-j} \in [x_{j}, 
\frac 34]) + \frac{1}{2^{n-j}}.
\end{equation}
Thus, by combining \eqref{n-i1} and \eqref{n-i2}, we obtain
\begin{align} \nonumber
&\text{Pr} (Z_{n} \in [x_{j+1},x_j]) \\ \label{last-2}
& \leq n \text{Pr} (Z_{n-j} \in [ x_{j}, \frac 34]) + \text{Pr}(Z_{n-j} 
\in [ x_{j},  x_j^{2^{-j}}]) + \frac{1}{2^{n-j}}.
\end{align}
Finally, in order to conclude the proof of \eqref{claim_i} (for $j \in \{2,\cdots,n\}$), we prove the following relations: 
\begin{equation} \label{claim_l1}
2^{-j}\text{Pr}(Z_{n-j} \in [ x_j, \frac 34]) \leq 
\text{Pr} (Z_n \in [\frac 14, \frac 34] ),   
\end{equation}
and 
\begin{equation} \label{claim_l2}
2^{-j} \text{Pr}(Z_{n-j} \in [ x_{j},  x_j^{2^{-j}}])\leq 
\text{Pr} (Z_n \in [\frac 14, \frac 34]).
\end{equation}
It is easy to see that these two relations combined with \eqref{last-2} will result in \eqref{claim_i}.
Firstly, note that for $j=2$ the relations \eqref{claim_l1} and \eqref{claim_l2} are trivial. Also, for $j \geq 3$ because of the fact that  $x_j^{2^{-j}} \geq \frac 34$, 
then \eqref{claim_l1} will be a direct consequence of \eqref{claim_l2}, and hence it is enough to prove 
\eqref{claim_l2}.

To prove  \eqref{claim_l2}, we show that
\begin{equation} \label{sufficient1}
 \text{Pr}(Z_n \in [\frac 14, \frac 34] \mid Z_{n-j} \in [ x_j, x_j^{2^{-j}}]) \geq 2^{-j}, 
\end{equation}
and from this we can conclude \eqref{claim_l2} by writing
\begin{align*}
 & \text{Pr}(Z_n \in [\frac 14, \frac 34]) \\
 & \geq   \text{Pr}(Z_n \in [\frac 14, \frac 34] \mid Z_{n-j} \in [ x_j, x_j^{2^{-j}}]) \times 
  \text{Pr}( Z_{n-j} \in [ x_j, x_j^{2^{-j}}]) \\
  & \geq 2^{-j} \text{Pr}( Z_{n-j} \in [ x_j, x_j^{2^{-j}}]). 
\end{align*}
It thus remains to show \eqref{sufficient1}. A moment of thought reveals that \eqref{sufficient1} is an immediate consequence of 
the following statement:
For any value $y$ s.t. $y \in [ x_j, x_j^{2^{-j}}]$,
there exists a $j$-tuple $(b_1, \cdots,b_j)\in \Omega_j$ 
such that $t_{b_1} \circ \cdots \circ t_{b_j} (y) \in  [\frac 14, \frac 34]$. 
We show this last statemet by constructing the binary values $b_1, \cdots, b_j$ in terms of $y$ (we use a similar approach as in Step 1). 
Consider the following algorithm: start with $y_0=y$ and for $1 \leq i \leq j$, we
recursively construct $b_{i}$ from $y_{i-1}$ by the following rule: 
If $t_0(y_{i-1}) \leq \frac 34$, then $b_i=0$ and $y_i=t_0(y_{i-1})$. Otherwise, let $b_i=1$ 
and $y_i=t_1(y_{i-1})$. To show that this algorithm succeeds in the sense that $y_j \in [\frac 14, \frac 34]$, we first observe that
once the value of $y_i$  
lies in the interval $[ \frac 14, \frac 34]$ (for some $1\leq i \leq j$), then
 for all $i \leq t \leq j$ we have $y_t \in [\frac 14 , \frac 34]$. 
Hence, we only need to show that by the above algorithm, the 
exists an index $i$, s.t. $1 \leq i \leq j$, and the 
value of $y_i$  lies in the interval $[\frac 14, \frac 34]$. On the one hand, assume 
$y \in [ x_j, \frac 14)$. We can then write
\begin{align*}
{\overbrace{t_0 \circ  \cdots \circ t_0}^{j \text{ times}}} (y)& =  1-(1-y)^{2^{j}}\\
& \geq  1-(1-x_j)^{2^{j}} \\
& \geq  \frac 12,
\end{align*}
where the last steps follows from the fact that $x_j=2^{-j}$.  
On the other hand,
assume $y \in (\frac 34, x_j^{2^{-j}}]$. 
We can write
\begin{align*}
{\overbrace{t_1 \circ  \cdots \circ t_1}^{j \text{ times}}}(y) & \leq (x_j^{2^{-j}})^{2^j}\\
& = x_j < \frac 34.
\end{align*}
As a result, the above algorithm always succeeds and the lemma is proved for $a=1-b=\frac 14$. \QED

\subsubsection{ Proof of Lemma~\ref{threshold}}
Recall that for a realization $\omega = \{b_k\}_{k \in \mathbb{N}} \in \Omega$ we define $\omega_n = (b_1, \cdots,b_n)$. The maps $t_0$ and $t_1$, hence the maps $\phi_{\omega_n}$, are strictly increasing maps on $[0,1]$. Thus
$\phi_{\omega_n}(z) \rightarrow 0$ implies that $\phi_{\omega_n}(z') \rightarrow 0$ for $z'
\leq z$ and $\phi_{\omega_n}(z) \rightarrow 1$ implies that $\phi_{\omega_n}(z') \rightarrow
1$ for $z' \geq z$. Moreover, we know that for almost every $z\in (0,1)$,
$\lim_{n \to \infty} \phi_{\omega_n} (z)$ is either $0$ or $1$ for almost every realization $\{\phi_{\omega_n}\}_{n \in \mathbb{N}}$. Hence, it suffices to let
\begin{equation}\nonumber
 z_{\omega}^*=\inf \{ z: \phi_{\omega_n}(z)\rightarrow 1 \}.
\end{equation}
To prove the second part of the lemma, notice that
\begin{align*}
z &= {\rm Pr}(Z_{\infty}=1)\\
&= {\rm Pr}(\phi_{\omega_n}(z)\rightarrow 1 )\\
&= {\rm Pr}(\inf \{ z: \phi_{\omega_n}(z)\rightarrow 1 \} \leq z ) \\
&= {\rm Pr}(z_{\omega}^* < z).
\end{align*}
Which shows that $z_{\omega}^*$ is uniformly distributed on $[0,1]$.

\subsubsection{Proof of Lemma~\ref{log_measure}}
In order to compute $\lim_{n \rightarrow \infty} \mathbb{E}[ \frac{1}{n}
\log (\phi_{\omega_n} ^{-1}(b)- \phi_{\omega_n} ^{-1}(a))] $,  we first define the process $\{ \bar{Z}_n \}_{n \in \mathbb{N}}$ with $\ \bar{Z}_0=z \in [0,1]$ and 
\begin{equation} \label{bZ_n}
 \bar{Z}_{n+1} =\left\{ \begin{array}{cc} 
\sqrt{ \bar{Z}_n},&\text{ w.p. }\frac12,\\
1-\sqrt{1- \bar{Z}_n},&\text{ w.p. }\frac12.
\end{array}\right.
\end{equation}
We can think of $ \bar{Z}_n$ as the reverse stochastic
process of $ Z_n$. Equivalently, we can also define $ \bar{Z}_n$  via the inverse
maps $t_0 ^{-1}$, $t_1 ^{-1}$. Consider the sequence of i.i.d. symmetric
Bernoulli random variables $B_1, B_2, \cdots$ and define
$\bar{Z}_n = \psi_{\omega_n} (z)$ where $\omega_n \triangleq (b_1, \cdots,b_n) \in \Omega_n$ and
\begin{equation}
 \psi_{\omega_n} = t_{b_n} ^{-1} \circ t_{b_{n-1}} ^{-1} \circ \cdots \circ t_{b_1} ^{-1}.
\end{equation}
We show that
the Lebesgue measure (or the uniform probability measure) on
$[0,1]$, denoted by $\nu$, is the unique, hence ergodic, invariant measure for the
Markov process $\bar{Z}_n$.
To prove this result, first note that if $\bar{Z}_n$ is distributed according to the Lebesgue measure, then
\begin{align*}
 {\rm Pr}( \bar{Z}_{n+1} < x)&=\frac{1}{2}{\rm Pr}( \bar{Z}_{n} < t_0(x))+ \frac{1}{2}{\rm Pr}( \bar{Z}_{n} < t_1(x))\\
&= \frac{1}{2}x^2 + \frac{1}{2} (2x-x^2)=x.
\end{align*}
Thus, $\bar{Z}_{n+1}$ is also  distributed according to the Lebesgue measure and  this implies the invariance of the Lebesgue measure for $\bar{Z}_{n}$. 
In order to prove
the uniqueness, we will show that for any $z \in (0,1)$, $\bar{Z}_n$
converges weakly to a uniformly distributed random point in $[0,1]$, i.e.,
\begin{equation} \label{conv_Lebesgue}
 \bar{Z}_n= \psi_{\omega_n}(z) \stackrel{d} {\rightarrow}\nu.
\end{equation}
Note that with \eqref{conv_Lebesgue} the uniqueness of $\nu$ is proved since for any invariant measure
$\rho$ assuming $\bar{Z}_n$ is distributed according to $\rho$, we have 
\begin{equation}
 \rho(\cdot)= {\rm Pr} (\bar{Z}_n \in  \cdot ) = \int {\rm Pr} (\bar{Z}_n \in \cdot) \rho (dz) \stackrel{d} {\rightarrow} \nu(\cdot).
\end{equation}
To prove \eqref{conv_Lebesgue}, note that $\psi_{\omega_n}$  has the same
(probability) law as $\phi_{\omega_n} ^{-1}$ and we know that $\phi_{\omega_n} ^{-1}(z)
\rightarrow z_{\omega}^*$ almost surely and hence weakly. Also, $z_{\omega}^*$ is distributed
according to $\nu$, which proves \eqref{conv_Lebesgue}.
We are now ready to show that 
\begin{equation} \label{s0}
 \lim_{n \rightarrow \infty} \mathbb{E}[ \frac{1}{n} 
 \log ( \phi_{\omega_n} ^{-1} (b) - \phi_{\omega_n} ^{-1} (a) ) ] = \frac {1}{2 \ln 2} - 1.
\end{equation}
Using the mean-value theorem, we can write
\begin{equation} \label{s1}
 \psi_{\omega_n}(a)-\psi_{\omega_n} (b) = \psi'_{\omega_n} (c)(b-a),
\end{equation}
for some $c \in (a,b)$.
And by chain rule,
\begin{small}
\begin{align*}
\psi'_{\omega_n} (c)&= (t_{b_n} ^{-1} \circ t_{b_{n-1}} ^{-1} \circ \cdots \circ t_{b_1} ^{-1})'(c)\\
&={t_{b_1} ^{-1}}'(c) \times {t_{b_2} ^{-1}}'(t_{b_1} ^{-1}(c)) \times \cdots \times {t_{b_n} ^{-1}}'(t_{b_{n-1}} ^{-1} \circ  \cdots \circ t_{b_1} ^{-1}(c))\\
&={t_{b_1} ^{-1}}'(\psi_{\omega_0} (c)) \times {t_{b_2} ^{-1}}'(\psi_{\omega_1} (c)) \times \cdots \times {t_{b_n} ^{-1}}'(\psi_{\omega_{n-1}}(c))),
\end{align*}
\end{small}
and after applying  $\frac{1}{n}\log(\cdot)$ to both sides we obtain
\begin{equation}\label{ergodic}
 \frac{1}{n} \log (\psi '_{\omega_n} (c))=\frac{1}{n} \sum_{j=1} ^{n} \ln {t_{b_j} ^{-1}}' (\psi_{\omega_{j-1}}(c)).
\end{equation}
By the ergodic theorem, the last expression  converges almost surely to the expectation of $\log {t_{B_1} ^{-1}}'(U)$,
where $U$ is assumed to be distributed according to $\nu$. Hence, the asymptotic value of \eqref{ergodic} can be
 computed as
\begin{align} \nonumber
\mathbb{E}& [\log {t_{B_1} ^{-1}}'(U)] \\ \nonumber
&= \frac{1}{2} \int_{0} ^{1} \log (\sqrt{x})' dx +\frac{1}{2} \int_{0} ^{1} \log (1-\sqrt{1-x})' dx \\
&=\frac {1}{2 \ln 2} - 1. \label{s2}
\end{align}
The proof now follows as a result of \eqref{s0}, \eqref{s1}, \eqref{ergodic}, and \eqref{s2}.

\subsubsection{Proof of Lemma~\ref{gamma}}
The proof is by contradiction. Let us assume the contrary, i.e., we assume there exists $n\geq m$ s.t.,
\begin{equation} \label{contr}
\text{Pr}(H_n \leq \alpha 2^{-n \theta}) > I(W)-\beta2^{-n\theta}.
\end{equation}
 In the following, we show that with such an assumption we reach to a contradiction.  We have
 \begin{align}\nonumber
 &\mathbb{E}[H_n(1-H_n)] \\ \nonumber
 &= \mathbb{E}[H_n(1-H_n) \mid H_n \leq \alpha 2^{-n\theta}] \text{Pr}(H_n \leq \alpha 2^{-n\theta}) \\
 & \; \; \: \; +  \mathbb{E}[H_n(1-H_n) \mid H_n > \alpha 2^{-n\theta}] \text{Pr}(H_n > \alpha 2^{-n\theta}). \label{cont1}
 \end{align}
 It is now easy to see that 
 \begin{align*}
 \mathbb{E}[H_n(1-H_n) \mid H_n \leq \alpha 2^{-n\theta}]  \leq \alpha 2^{-n\theta},
\end{align*}
and since $\mathbb{E}[H_n(1-H_n)] \geq \gamma 2^{-n\theta}$, by using \eqref{cont1} we get
 \begin{equation} \label{cont2}
 \mathbb{E}[H_n(1-H_n) \mid H_n > \alpha 2^{-n\theta}] \text{Pr}(H_n > \alpha 2^{-n\theta}) \geq 2^{-n\theta}(\gamma-\alpha).
 \end{equation}
  We can further write
  \begin{align}\nonumber
 \mathbb{E}[(1-H_n)] &= \mathbb{E}[1-H_n \mid H_n \leq \alpha 2^{-n\theta}] \text{Pr}(H_n \leq \alpha 2^{-n\theta}) \\
 & \; \; \: \; +  \mathbb{E}[1-H_n \mid H_n > \alpha 2^{-n\theta}] \text{Pr}(H_n > \alpha 2^{-n\theta}), \label{cont3}
 \end{align}
 and noticing fact that $1-H_n \geq H_n (1-H_n)$ we can plug \eqref{cont2} in \eqref{cont3} to obtain
 \begin{align} 
 \mathbb{E}[(1-H_n)]  & \geq  \mathbb{E}[1-H_n \mid H_n \leq \alpha 2^{-n\theta}] \text{Pr}(H_n \leq \alpha 2^{-n\theta}) \nonumber \\ 
  & \;\;\; +2^{-n\theta}(\gamma-\alpha). \label{cont4}
 \end{align}
 We now continue by using \eqref{contr} in \eqref{cont4} to obtain
 \begin{align*}
 \mathbb{E}[(1-H_n)]  &> (1-\alpha2^{-n\theta})(I(W)-\beta2^{-n\theta})+ 2^{-n\theta}(\gamma-\alpha) \\
 &\geq I(W)+2^{-n \theta} (\gamma- \alpha (1+ I(W))-\beta),
 \end{align*} 
 and since $2\alpha+ \beta=\gamma$, we get $\mathbb{E}[1-H_n]>I(W)$. This is a contradiction
 since $H_n$ is a martingale and $\mathbb{E}[1-H_n]=I(W)$.

\section{Auxiliary Lemmas}
\begin{lemma} \label{1aux1}
Consider a channel $W$ with its Bhattacharyya process $Z_n=Z(W_n)$ and assume that for $n \in \mathbb{N}$
\begin{equation}\label{1auxlineq}
\mathbb{E}[(Z_n(1-Z_n))^\alpha]  \leq \beta 2^{-n \rho},
\end{equation}
where $\alpha, \beta, \rho$ are positive constants with $\alpha<1$. 
We then have for $n \in \mathbb{N}$
\begin{equation}
{\rm{Pr}} ( Z_n \leq \frac 12) \geq I(W) - c_1 2^{-n \rho}, 
\end{equation}
where $c_1$ is a positive constant that depends on $\alpha, \beta, \rho$.
\end{lemma} 
\begin{proof}
The proof consists of three steps. 
First, consider an arbitrary BMS  channel $W$ and let $Z_n=Z(W_n)$. 
Also, consider the process $Y_n = 1-Z_n^2$ . By using the relations \eqref{Zgen-} and \eqref{Zgen+}, it can easily be
checked that the process $Y_n$
has the form of \eqref{1genericX_n} and hence Lemma~\ref{1aux2} is applicable to $Y_n$. We thus have from \eqref{1aux2main}
that for $n \in \mathbb{N}$
\begin{equation*}
 \text{Pr}(Y_n > \frac 12) \leq c_2 Y_0 (1 + \log \frac{1}{ Y_0} ).
\end{equation*}
As a consequence 
\begin{align}
I(W) &= \lim_{n \to \infty}  \text{Pr}(Y_n > \frac 12) \nonumber\\
& \leq c_2 (1-Z(W)^2)(1 + \log \frac{1}{ 1-Z(W)^2} ). \label{1HZ}
\end{align}
In the second step, we consider a channel $W$ for which \eqref{1auxlineq} holds for $n \in \mathbb{N}$.
 By using \eqref{1auxlineq}, it is easy to see that for $n \in \mathbb{N}$
 \begin{align} \nonumber
 & \mathbb{E}[(Z_n^2(1-Z_n^2))^\alpha \,\, \mathbbm{1}_{\{Z_n > \frac 12 \}}] \\ \nonumber
 & =  \mathbb{E}[(Z_n (1+Z_n))^\alpha (Z_n(1-Z_n))^\alpha \mathbbm{1}_{\{Z_n > \frac 12 \}}] \\ \nonumber
 & \leq \sup_{z \in [\frac 12, 1]} (z (1+z))^\alpha \,\,  \mathbb{E}[(Z_n(1-Z_n))^\alpha \mathbbm{1}_{\{Z_n > \frac 12\}}] \\
 & \leq 2^\alpha \beta 2^{-n \rho} \leq \beta 2^{1-n \rho}. \label{11ineq}
 \end{align}
In the final step, we consider a number $n \in \mathbb{N}$ and let $N=2^n$. We then define the set $\mathcal{A}$ as
\begin{equation*}
\mathcal{A} = \{ i \in \{0,1, \cdots,N-1\}: Z(W_N^{(i)}) \leq \frac 12\},
\end{equation*} 
with $\mathcal{A}^c$ being its complement.  We have
\begin{align*}
& \sum_{i \in \mathcal{A}^c} I(W_N^{(i)}) \\
&\stackrel{(a)}  { \leq}\sum_{i \in \mathcal{A}^c} c_2 (1-Z(W_N^{(i)})^2)(1 + \log \frac{1}{ 1-Z(W_N^{(i)})^2} ) \\
& \stackrel{(b)}{\leq}  \sum_{i \in \mathcal{A}^c} 4c_2 c_3 \bigl(Z(W_N^{(i)})^2(1-Z(W_N^{(i)})^2) \bigr )^ \alpha \\
&= 4 c_2 c_3N\mathbb{E}[(Z_n^2(1-Z_n^2))^\alpha \,\, \mathbbm{1}_{\{Z_n > \frac 12 \}}] \\
&\stackrel{(c)}{\leq}  8c_2c_3 N \beta 2^{-n \rho}.
\end{align*}
Here (a) follows from  \eqref{1HZ}, (b) follows from Lemma~\ref{1aux3} and the fact that for $x \leq \frac 34$ we have $1+\log \frac{1}{x} \leq 4 \log \frac{1}{x}$, and (c) follows from \eqref{11ineq}.
Now, as a consequence of the above chain of inequalities we have
\begin{align*}
|\mathcal{A}| &\geq  \sum_{i \in \mathcal{A}} I(W_N^{(i)})\\
 & =  NI(W) - \sum_{i \in \mathcal{A}^c} I(W_N^{(i)}) \\ 
& \geq N(I(W) - 8c_2c_3\beta2^{-n \rho}),
\end{align*}
and consequently
\begin{equation*}
{\rm Pr}(Z_n \leq \frac 12) =\frac{|\mathcal{A}|}{N}  \geq I(W) - 8c_2c_3\beta2^{-n \rho}.
\end{equation*}
Hence, the proof follows by letting $c_1 = 8c_2c_3\beta$. 
\end{proof}


\begin{lemma}\label{1aux2}
Consider a generic stochastic process $\{X_n\}_{n\geq0}$ s.t. $X_0=x$, where $x \in (0,1)$, and for $n\geq 1$
\begin{equation} \label{1genericX_n}
X_{n} \leq \left\{
\begin{array}{lr}
 X_{n-1}^2 &  ; \text{if } B_n=1,\\
 2X_{n-1} &   ; \text{if } B_n=0.
\end{array} \right.
\end{equation}
Here,  $\{B_n\}_{n\geq 1}$ is a sequence of iid random variables with distribution  Bernoulli($\frac 12$). We then have for $n \in \mathbb{N}$
\begin{equation} \label{1aux2main}
{\rm Pr}(X_n \leq 2^{-2^{\sum_{i=1}^{n} B_i}}) \geq 1- c_2 x (1+ \log \frac {1}{x}), 
\end{equation}
where $c_2$ is a positive constant.
\end{lemma} 
\begin{proof}
We begin by recalling some related notation. 
Assuming $\{B_n\}_{n \in \naturals}$ is a sequence of iid Bernoulli($\frac12$) random variables, 
we denote by $(\mathcal{F}, \Omega,  \text{Pr})$ the probability space generated by this 
sequence. We also let  $(\mathcal{F}_n, \Omega_n, \text{Pr}_n)$ be the probability space 
generated by $(B_1, \cdots,B_n)$.  Finally, we denote by $\theta_n$ the natural embedding 
of $\mathcal{F}_n$ into  $\mathcal{F}$, i.e., for every $F \in \mathcal{F}_n$
\begin{equation*}
\theta_n(F) =\{ (b_1, b_2,\cdots, b_n, b_{n+1}, \cdots) \in \Omega \mid (b_1, \cdots, b_n ) \in F \}.
\end{equation*}
We thus have $ \text{Pr}_n (F) =  \text{Pr}(\theta_n (F))$.

We slightly modify  $X_n$ to start with
 $X_0=x$, where $x \in (0,1)$, and for $n\geq 1$
\begin{equation} 
X_{n} = \left\{
\begin{array}{lr}
 X_{n-1}^2 &  ; \text{if } B_n=1,\\
 2X_{n-1} &   ; \text{if } B_n=0.
\end{array} \right.
\end{equation}
It is easy to see that if we prove the lemma for this version of $X_n$, then the result of the lemma is valid for 
any generic $X_n$ that satisfies \eqref{1genericX_n}.

Equivalently, we can analyze the process 
$A_n= -\log X_n$ , i.e., $A_0=- \log  x \triangleq a_0$ and
\begin{equation} \label{A}
A_{n+1}=  \left\{
\begin{array}{lr}
2{A_n} &  ; \text{if } B_n=1,\\
A_n - 1   &  ; \text{if } B_n=0.
\end{array} \right.
\end{equation}
Note that in terms of the process $A_n$, the statement of the lemma can be phrased as
\begin{equation} \label{P(A_n)}
 {\rm{Pr}} (A_n \geq  2^{\sum_{i=1} ^{n} B_i}) \geq 1- c_2 \frac{1+a_0}{2^{a_0}}.
\end{equation}
Let us first explain how to associate to each $ (b_1, \cdots, b_n) \triangleq \omega_n  \in
\Omega_n$ a sequence of ``runs" $(r_1, \cdots , r_{k(\omega_n)})$. This sequence
is constructed by the following procedure. Each of the $r_i$'s is a positive integer. We construct the integers $r_i$ one by one starting from $r_1$.  We define $r_1$ as the smallest index
$i \in \mathbb{N}$ so that $b_{i+1} \neq b_1$. In general, $r_k$ is constructed from the previous $r_i$'s, $1 \leq i < k$, in the following way. 
If $\sum_{j=1}
^{k-1} r_j < n$ then
\begin{align*}
r_k = \min \{i \mid  \sum_{j=1} ^{k-1} r_j < i \leq n , b_{i+1} \neq b_{\sum_{j=1} ^{k-1} r_j}\}-\sum_{j=1} ^{k-1} r_j.
\end{align*}
The process stops whenever the sum of the runs equals $n$ (i.e., whenever $\sum_{i=1}^k r_i$ is exactly equal to $n$). Denote the
stopping time of the process by $k(\omega_n)$. In words, the sequence
$(b_1,\cdots, b_n)$ starts with $b_1$. It then repeats $b_1$, $r_1$ times.
Next follow $r_2$ instances of $\overline{b_1}$ ($\overline{b_1}:=1-b_1$), followed again by $r_3$
instances of $b_1$, and so on. We see that $b_1$ and $(r_1, \cdots,
r_{k(\omega_n)})$ fully describe $\omega_n =(b_1,\cdots, b_n) $. Therefore,
there is a one-to-one map
\begin{equation} \label{one-to-one}
 (b_1,\cdots, b_n) \longleftrightarrow \{ b_1, (r_1, \cdots, r_{k(\omega_n)})\}.
\end{equation}
As an example, for the sequence $\omega_{8} \triangleq (b_1, b_2, \cdots, b_8) = (1, 0, 0, 1, 0,0,0,1)$, we have $k(\omega_8) = 5$, and the corresponding sequence of runs is $(r_1, r_2, r_3, r_4, r_5) = (1,2,1,3,1)$. Also, the knowledge of the sequence $(r_1, r_2, r_3, r_4, r_5)$ and the fact that $b_1 = 1$ will uniquely determine the sequence $(b_1, b_2, \cdots, b_8)$.  

We think of $\omega_n = (b_1, \cdots, b_n)$ as a realization of the random vector$(B_1, \cdots, B_n)$. In this regard, each realisation $(b_1,\cdots, b_n )$ is associated with a value $k(w_n)$ and a run sequence $(r_1, \cdots, r_{k(\omega_n)})$. Thus,  $k(\omega_n)$ and $(r_1, \cdots, r_{k(\omega_n)})$ are similarly the corresponding realizations of random objects which we denote by $K$ and $(R_1, \cdots, R_K)$. 

Note that for a generic sequence $(b_1, \cdots, b_n)$ we can either have $b_1 = 1$ or $b_1 = 0$. We start with the first case, i.e., we first condition ourselves on the event $B_1 = 1$. 

\underline{Case I ($b_1 = 1$):} 
It is easy to see that assuming $b_1  = 1$ we have:
\begin{equation} \label{l1}
\sum_{i=1}^{n} b_i = \sum_{ \text{$j$ odd $\leq k(\omega_n)$}} r_j,
\end{equation}
and
\begin{equation} \label{l2} 
 n = \sum_{j=1} ^{k(\omega_n)} r_j.
\end{equation}
Analogously, for a realization $(b_1,b_2, \cdots) \triangleq \omega \in \Omega$ of the infinite sequence of random variable $\{B_i\} _{i \in \mathbb{N}}$, we can associate a sequence of runs $(r_1,r_2, \cdots)$. In this regard, considering the infinite sequence of random variables $\{B_i\} _{i \in  \mathbb{N}}$ (with the extra condition $B_1 =1$), the corresponding sequence of runs, which we denote by $\{ R_k \}_{k \in \mathbb{N}}$, is an iid sequence with ${\rm{Pr}}(R_i= j) = \frac{1}{2^j}$.
Let us now see how we can express the output of $A_n$ in terms of the runs $r_1,
r_2, \cdots, r_{k(\omega_n)}$.  We begin by a simple example: Consider
a sequence $(b_1=1, b_2, \cdots,b_8)$ that has an associated run sequence
$(r_1,\cdots, r_5)=(1,2,1,3,1)$.  For such a choice of $b_i$'s,  we will now write the value of the process $A_n$ for several small values of $n$. 
In this way,  it is easy to notice a simple pattern for the evolution of $A_n$ in terms of the sequence of runs. We have  
\begin{align*}
A_{1}    &= a_0 2^{r_1} , \\
A_{3}    &= a_0 2^{r_1}  - r_2, \\
A_{4}    &= (a_0 2^{r_1}  - r_2) 2^{r_3}= a_0 2^{r_1+r_3}  - r_2 2^{r_3}  , \\
A_{7}    &= (a_0 2^{r_1}  - r_2) 2^{r_3} - r_4= a_0 2^{r_1+r_3}  - r_2 2^{r_3} - r_4  ,\\
A_{8}    &= ((a_0 \times 2^{r_1} - r_2) \times 2^{r_3} -r_4) \times 2^{r_5} \\
         & = a_0 2^{r_1+r_3+r_5} -r_2 2^{r_3+r_5} - r_4 2^{r_5} \\
         & = 2^{r_1+r_3+r_5}( a_0 - 2^{-r_1} r_2- 2^{-(r_1+r_3)} r_4).
\end{align*}
In general, for a sequence $(b_1, \cdots, b_n)$ with the associated run
sequence $(r_1,\cdots, r_{k(\omega_n)})$ we can write (note that $b_1 = 1$):
\begin{align} \nonumber
A_{n} & = a_0 2^{\sum_{\text{$i$ odd $\leq k(\omega_n)$}} r_i} - \!  \!  \!  \!  \!  \!  \! \sum_{\text{$i$ even $\leq k(\omega_n)$}} \!  \!  \!  \!  \!  \!  \! r_i  2 ^{\sum_{\text{$j$ odd $, i< j \leq k(\omega_n)$}} r_j} \\ \nonumber
&= a_0 2^{\sum_{\text{$i$ odd $\leq k(\omega_n)$}} r_i} -\!  \!  \!  \!  \!  \!  \!  \!  \!  \!  \!  \!  \sum_{\!  \text{$i$ even $\leq k(\omega_n)$}} \!  \!  \!  \!  \!  \!  \!  \!   r_i  2 ^{(-\sum_{\text{$j$ odd $< i$}} r_j + \sum_{\text{$j$ odd $\leq k(\omega_n)$}} r_j )} \\ \nonumber
&= [2^{\sum_{\text{$i$ odd $\leq k(\omega_n)$}} r_i}][a_0 -  (\!  \!  \!  \!  \!  \!  \! \sum_{\text{$i$ even $\leq k(\omega_n)$}} \!  \!  \!  \!  \!  \!  \! r_i  2 ^{-\sum_{\text{$j$ odd $< i$}} r_j} ) ] \\ \label{A_n_eq}
&\stackrel{\eqref{l1}}{=} [2^{\sum_{i=1} ^{n} b_i}] [a_0 -  (\!  \!  \!  \!  \!  \!  \! \sum_{\text{$i$ even $\leq k(\omega_n)$}} \!  \!  \!  \!  \!  \!  \! r_i  2 ^{-\sum_{\text{$j$ odd $< i$}} r_j }) ].
\end{align}
Here, by $\sum_{\text{$i$ even $\leq j$}}$ we mean that the sum is over all the positive integers $i$ that are even and are also less than the given value $j$.  
Similarly, for example by $\sum_{\text{$i$ odd: $j < i \leq k(\omega_n)$}}$ we mean that the sum is over all 
integers $i$  that are odd and also satisfy $j < i \leq k(\omega_n)$.     
Now, if we consider the random vector $(B_1, B_2, \cdots, B_n)$ and its associated run sequence $(R_1, \cdots, R_K)$, we can write
\begin{align} \nonumber
& {\rm{Pr}} (A_n \geq  2^{\sum_{i=1} ^{n} B_i})\\ \nonumber
&\stackrel{\eqref{A_n_eq}}{=} {\rm{Pr}}  \bigl (2^{\sum_{i=1}^n B_i} (a_0 -  \!\! \sum_{\text{$i$ even $\leq K$}} \!\! R_i  2 ^{-\sum_{\text{$j$ odd $< i$}} R_j} ) \geq 2^{\sum_{i=1}^n B_i} \bigr ) \\ \label{eql1}
&= {\rm{Pr}}  (a_0 -  \sum_{\text{$i$ even $\leq K$}} R_i  2 ^{-\sum_{\text{$j$ odd $< i$}} R_j} \geq 1 ),
\end{align}
Our objective is to find a lower bound on the left-hand side of \eqref{P(A_n)}. In this regard, by using \eqref{eql1}, we can equivalently  find an upper-bound on the probability of the complementary event:
\begin{equation}  \label{prob}
{\rm{Pr}} (\sum_{\text{$i$ even $\leq K$}} R_i  2 ^{-\sum_{\text{$j$ odd $< i$}} R_j} > a_0-1).
\end{equation}
For $n \in \mathbb{N}$, define the set $U_n \in \mathcal{F}_n$ as
\begin{equation*}
U_n =  \{ \omega_n \in \Omega_n \mid  \exists \, l \leq k(\omega_n) : \sum_{\text{$i$ even $\leq l$}} r_i  2 ^{-\sum_{\text{$j$ odd $< i$}} r_j} \geq a_0-1   \}.
\end{equation*}
Clearly we have:
\begin{equation*}
{\rm{Pr}}  (\sum_{\text{$i$ even $\leq K$}} R_i  2 ^{-\sum_{\text{$j$ odd $< i$}} R_j} \geq a_0-1) \leq {\rm{Pr}}(U_n).
\end{equation*}
Obtaining an upper bound on $\text{Pr}(U_n)$ for finite $n$ seems to be a difficult task. This is because for finite $n$ handling the distribution of the runs is cumbersome. The idea here is to show that we can obtain useful bounds on $\text{Pr}(U_n)$ (for any finite $n$) by considering the case when $n$ tends to $\infty$. In the infinite $n$ limit, the run sequence $\{R_i\}_{i \in \naturals}$ becomes an iid sequence (note that $B_1=1$) and this makes the proofs much simpler.  

In the following we show that if   $(b_1, \cdots,b_n) \in U_{n} $, then for any choice of $b_{n+1}$, it is true that $(b_1, \cdots,b_n,b_{n+1}) \in U_{n+1}$. The two bits $b_n$ and $b_{n+1}$ and can jointly take four possible values. Here, for the sake of brevity, we will only consider the case when $b_n,b_{n+1}=1$, and the other three cases can be verified similarly. Let $\omega_n = (b_1, \cdots, b_{n-1},b_n=1) \in U_n$. Hence, $k(\omega_n)$ is an odd number (recall that $b_1=1$) and the quantity $\sum_{\text{$i$ even $\leq k(\omega_n)$}} r_i  2 ^{-\sum_{\text{$j$ odd $< i$}} r_j}$ does not depend on the value of $r_{k(\omega_n)}$. Now consider the sequence  $\omega_{n+1}=(b_1, \cdots,b_n=1, 1)$. Since the last bit ($b_{n+1}$) equals $1$, then $k(\omega_{n+1})= k(\omega_n)$ (i.e. the two sequences $\omega_n$ and $\omega_{n+1}$ have the same number of runs). Therefore, it is easy to see that 
\begin{equation*}
\sum_{\text{$i$ even $\leq k(\omega_n)$}} r_i  2 ^{-\sum_{\text{$j$ odd $< i$}} r_j}
=
\sum_{\text{$i$ even $\leq k(\omega_{n+1})$}} r_i  2 ^{-\sum_{\text{$j$ odd $< i$}} r_j}.
\end{equation*}
 As a result $(b_1, \cdots,b_n,1) \in U_{n+1}$.
From above, we conclude that for any $i \in \naturals$ we have $\theta_i(U_i) \subseteq \theta_{i+1}(U_{i+1})$ and as a result
\begin{align*}
{\rm{Pr}}_i (U_i) &= {\rm{Pr}}(\theta_i(U_i)) \leq {\rm{Pr}}(\theta_{i+1}(U_{i+1})) ={\rm{Pr}}_{i+1}(U_{i+1}).
\end{align*}
Hence, the quantity $\lim_{n \to \infty} {\rm{Pr}}_n (U_n) =\lim_{n \to \infty} {\rm{Pr}}  (\theta_n(U_n))=\lim_{n \to \infty} {\rm{Pr}}(\cup_{i=1}^{n} \theta_i (U_i))$ is an upper bound on \eqref{prob}.
Let us now consider the set\footnote{Note here that the $V \subseteq \Omega$, while $U_n \subseteq \Omega_n$.}
\begin{equation*}
V =  \{ \omega \in \Omega \mid \exists \, l  \in \naturals : \sum_{\text{$i$ even $\leq l$}} r_i  2 ^{-\sum_{\text{$j$ odd $< i$}} r_j} \geq a_0 -1 \}.
\end{equation*}
By the definition of  $V$ we  have $ \cup_{i=1} ^{\infty} \theta_i (U_i) \subseteq V$, and as a result,  ${\rm{Pr}} (\cup_{i=1} ^{\infty} \theta_i (U_i)) \leq {\rm{Pr}}(V)$.   In order to bound the probability of the set $V$, note that assuming $B_1=1$,  the sequence $\{ R_k \}_{k \in \mathbb{N}}$ (i.e., the sequence of runs when associated with the sequence $\{B_i\}_{i \in \mathbb{N}}$) is an iid sequence with ${\rm{Pr}}(R_i= j) = \frac{1}{2^j}$.  We also have
\begin{align} \nonumber
\text{Pr}(V) &=  {\rm{Pr}} (a_0 -  \sum_{\text{$i$ even}} R_i  2 ^{-\sum_{\text{$j$ odd $ < i$}} R_j} \leq 1 )\\
& = {\rm{Pr}}  (\sum_{\text{$i$ even}} R_i  2 ^{-\sum_{\text{$j$ odd $ < i$}} R_j} \geq  a_0 -1) \nonumber \\
& = {\rm{Pr}}  (2^{\sum_{\text{$i$ even}} R_i  2 ^{-\sum_{\text{$j$ odd $ < i$}} R_j}} \geq  2^{a_0 -1})\nonumber \\ 
& \leq \frac{\mathbb{E}[2^{\sum_{\text{$i$ even }} R_i  2 ^{-\sum_{\text{$j$ odd $ < i$}} R_j }}]}{2^{a_0-1}}, \label{markov}
\end{align}
where the last step follows from the Markov inequality. The idea is now to provide an upper bound on the quantity $\mathbb{E}[2^{\sum_{\text{$i$ even}} R_i  2 ^{-\sum_{\text{$j$ odd $ < i$}} R_j }}]$. Let $X=\sum_{\text{$i$ even }} R_i  2 ^{-\sum_{\text{$j$ odd $ < i$}} R_j }$. We have
\begin{align} \nonumber
 \mathbb{E}[2^X] 
& =\sum_{s=1}^{\infty} {\rm{Pr}}(R_1 = s) \mathbb{E}[2^X \mid R_1= s]\\ \nonumber
& = \sum_{s=1}^{\infty} \frac{1}{2^{s}} \mathbb{E}[2^X \mid R_1= s]\\ \nonumber
& \stackrel{(a)}{=} \sum_{s=1}^{\infty} \frac{1}{2^{s}}  \mathbb{E}[2^{\frac{R_2}{2^{s}}}]\mathbb{E}[2^{\frac{X}{2^{s}}}]\\ \nonumber
& \stackrel{(b)}{=} \sum_{s=1}^{\infty} \frac{1}{2^{s}( 2^{1-\frac{1}{2^s}} -1)} \mathbb{E}[2^{\frac{X}{2^{s}}}]\\ \label{l3}
& \stackrel{(c)}{\leq} \sum_{s=1}^{\infty} \frac{1}{2^{s}( 2^{1-\frac{1}{2^s}} -1)}  (\mathbb{E}[2^X])^{\frac{1}{2^s}},
\end{align}
where (a) follows from the fact that $R_i$'s are iid and $X$ is self-similar, (b) follows from the relation $\mathbb{E}[2^{\frac{R_2}{2^s}}] = \frac{1}{2^{1 - 2^{-s}} -1}$, and (c) follows from Jensen inequality. Now, because $X$ is a positive random variable we have $\mathbb{E}[2^X] \geq 1$, and consequently the sequence $\{(\mathbb{E}[2^X])^{2^{-s}}\}_{s \in \mathbb{N}}$ is a decreasing sequence.
As a result, from \eqref{l3} an upper bound on the the value of $\mathbb{E}[2^X]$ can be derived as follows. We have
\begin{align*}
\mathbb{E}[2^X] & \leq \frac{1}{2(2^{\frac 12} -1) }( \mathbb{E}[2^X] )^{\frac 12} (1 + \frac 12 + \frac 14 + \cdots)\\   
& \leq \frac{1}{(2^{\frac 12} -1) }( \mathbb{E}[2^X] )^{\frac 12},
\end{align*}
and as a result, 
we have 
\begin{equation} \label{tild_c_2}
\mathbb{E}[2^X] \leq \frac{1}{(2^{\frac 12} - 1)^2} \triangleq \tilde{c}.
\end{equation}
Thus, by \eqref{markov} we obtain 
\begin{align*}
{\rm{Pr}} (a_0 -  \sum_{\text{$i$ even $\leq m$}} R_i  2 ^{-\sum_{\text{$j$ odd $ < i$}} R_j} \leq 1 ) \leq \frac{\tilde{c}}{2^{a_0 - 1}}.
\end{align*}
Thus, given that $B_1=1$, we have:
\begin{equation*}
 {\rm{Pr}} (A_n \geq  2^{\sum_{i=1} ^{n} B_i}) \geq 1- \frac{\tilde{c}}{2^{a_0-1}}.
\end{equation*}
Or more precisely we have
\begin{equation} \label{B_1=1-final}
 {\rm{Pr}} (A_n \geq  2^{\sum_{i=1} ^{n} B_i} \mid B_1=1, A_0 = a_0) \geq 1- \frac{\tilde{c}}{2^{a_0-1}}.
\end{equation}
\underline{Case II ($b_1 =0$):} Now consider the case $b_1 = 0$ (i.e., we condition on the event $B_1=0$). We show that a similar bound applies for $A_n$. Firstly, note that by fixing the value of $n$ the distribution of $R_1$ is as follows: ${\rm{Pr}}(R_1=i | B_1 =0)=\frac{1}{2^i}$ for $1 \leq i \leq n-1$ and ${\rm{Pr}}(R_1=n | B_1 =0) = \frac{1}{2^{n-1}}$. We have
\begin{small}
\begin{align} \nonumber
 & {\rm{Pr}} (A_n \geq  2^{\sum_{i=1} ^{n} B_i} \mid B_1=0)\\ \nonumber
& = \sum_{j=1}^{n} {\rm{Pr}} (A_n \geq  2^{\sum_{i=1} ^{n} B_i} \mid R_1 = j, B_1=0) {\rm{Pr}}(R_1 =j \mid B_1=0)\\ \nonumber
& \leq  \sum_{j =1}^{\text{min}( a_0 -1,  n)}  \!\!\!\! {\rm{Pr}} (A_n \geq  2^{\sum_{i=1} ^{n} B_i}\mid R_1 = j, B_1=0){\rm{Pr}}(R_1 =j \mid B_1=0)\\ \nonumber
& \,\,\, \,\,\, \,\,\, \,\, \,\,\, \,\,\, \,\,\,  + \sum_{j > a_0 -1}^n {\rm{Pr}}(R_1=j \mid B_1 =0)\\ \nonumber
&\stackrel{\eqref{A}}{\leq}  \sum_{j =1}^{\text{min}( a_0 -1,  n)}  {\rm{Pr}} (A_{n-j} \geq  2^{\sum_{i=j+1} ^{n} B_i}\mid  B_{j+1}=1, A_{j} = a_0 - j ) \times \frac{1}{2^j} \\ \nonumber
&\,\,\, \,\,\, \,\,\, \,\, \,\,\, \,\,\, \,\,\, +  2\sum_{j > a_0 -1}^\infty  \frac{1}{2^{j}} \\ \nonumber
&\stackrel{\eqref{B_1=1-final}}{\leq}  \sum_{j =1}^{\text{min}( a_0 -1,  n)} \frac{1}{2^j} \frac{\tilde{c}}{2^{a_0 - 1 -j}} + \frac{4}{2^{a_0-1}}\\ \nonumber
&\leq \frac{\tilde{c} a_0 +4 }{2^{a_0-1}} .
\end{align}
\end{small}
Thus, we can write (note by \eqref{tild_c_2} that $\tilde{c} > 4$)
\begin{equation} \label{B_1=0-final}
 {\rm{Pr}} (A_n \geq  2^{\sum_{i=1} ^{n} B_i} \mid B_1=0, A_0 = a_0) \geq 1- \frac{\tilde{c} (1+a_0)}{2^{a_0-1}}.
\end{equation}

Finally, by considering the two cases together, we obtain from \eqref{B_1=1-final} and \eqref{B_1=0-final} the following:
\begin{equation*}
 {\rm{Pr}} (A_n \geq  2^{\sum_{i=1}^{n} B_i}) \geq 1- \frac{2\tilde{c}(1+a_0)}{2^{a_0}}.
\end{equation*}
Hence, the proof of the lemma  follows with $c_2 =2 \tilde{c} = \frac{2}{(2^\frac 12 -1)^2}$.
\end{proof}


\begin{lemma}  \label{1aux3}
Let $\alpha<1$ be a constant. We have for $x \in (0,\frac{3}{4}]$
\begin{equation}\label{1ineqlog}
x \log(\frac{1}{x}) \leq c_3 (x(1-x))^\alpha,
\end{equation}
where 
\begin{equation} \label{1lemc3}
c_3=  \frac{2}{(1-\alpha)\ln 2}.
\end{equation}
\end{lemma}
\begin{proof}
By applying the function $\log(\cdot)$ to both sides of \eqref{1ineqlog} and some further simplifications, the inequality \eqref{1ineqlog} is equivalent to the following:
For $x \in (0, \frac 34]$
\begin{equation*}
\log(\log\frac {1}{x}) \leq \log c_3 + (1-\alpha) \log \frac{1}{x} + \alpha \log(1-x).
\end{equation*}  
As $x \leq \frac 34$, we have $\alpha \log(1-x) \geq -\log 4$. Hence, in order for the above inequality to hold it is sufficient that for $x \in (0, \frac 34]$
\begin{equation*}
\log(\log \frac {1}{x} ) \leq \log \frac{c_3}{4} + (1-\alpha) \log \frac{1}{x}.
\end{equation*}  
Now, by letting $u = \log \frac {1}{x}$, the last inequality becomes
\begin{equation}
(1-\alpha)u- \log u+ \log \frac{c_3}{4} \geq 0,
\end{equation}
for $u \geq \log(\frac 43)$. It is now easy to check that by the choice of $c_3$ as in \eqref{1lemc3}, the minimum of the above expression over the range $u \geq  \log(\frac 43)$ is always non-negative and hence the proof follows. 
\end{proof}

\end{appendices}

\end{document}